\documentclass{article}

\usepackage[utf8]{inputenc}
\usepackage{xspace}
\usepackage{setspace}

\usepackage{amsmath, amssymb, amsfonts}
\usepackage[only, lightning, llbracket, rrbracket]{stmaryrd}
\usepackage{mathtools} 
\usepackage{tikz}
\usepackage{multirow}
\usepackage{times}
\usepackage{url}
\usepackage{hyperref}
\usetikzlibrary{petri,positioning}

\newtheorem{theorem}{Theorem}
\newtheorem{lemma}{Lemma}
\newtheorem{proposition}{Proposition}
\newtheorem{corollary}{Corollary}
\newtheorem{definition}{Definition}
\newtheorem{example}{Example}
\newcommand{\qed}{\hfill$\Box$}

\usepackage[textsize=footnotesize,color=blue!40,prependcaption,disable]{todonotes}

\newcommand{\treename}{Karp\& Miller\xspace}
\newcommand{\treeshortname}{{\sf KM}\xspace}
\newenvironment{proof}{\noindent{\it Proof.\ \ }}{}

\definecolor{darkviolet}{rgb}{0.5,0,0.4}
\definecolor{darkgreen}{rgb}{0,0.4,0.2}
\definecolor{darkblue}{rgb}{0.1,0.1,0.9}
\definecolor{darkgrey}{rgb}{0.5,0.5,0.5}
\definecolor{lightblue}{rgb}{0.4,0.4,1}
\usepackage{listings}
\usepackage{adjustbox} 
\lstnewenvironment{pseudo}[1][]{
\lstset{
    basicstyle=\small\ttfamily,
    emphstyle=\color{red}\bfseries,
    keywordstyle=\color{darkviolet}\bfseries,
    commentstyle=\color{darkgreen},
    stringstyle=\color{darkblue},
    numberstyle=\color{darkgrey}\ttfamily,
    emphstyle=\color{red},
    backgroundcolor=\color{white},
    morecomment=[s][\color{lightblue}]{/*}{*/},
    morecomment=[l]{\\},
  escapechar=@,
  showstringspaces=false,
  emptylines=0,
  aboveskip=0ex,
  belowskip=0ex,
  numbers=left,
  mathescape=true,
otherkeywords={-, =, +, [, ], (, ), \{, \}, :,::, *, !,;,&,|,?},
morekeywords={while,initialize, forall,if, elif,else, with, where,
st, in, and, or,return, post, int, for, step},
emph={},
emph={[2]OK, KO, \?},
emphstyle=[2]\color{darkblue},
emph={[3]subroutine},
emphstyle=[3]\color{darkgreen},
}}{}

\newcommand{\wPN}{\ensuremath{\omega}PN\xspace}
\newcommand{\wPNR}{\ensuremath{\omega}PN+R\xspace}
\newcommand{\PNT}{PN+T\xspace}
\newcommand{\PNR}{PN+R\xspace}
\newcommand{\wPNT}{\ensuremath{\omega}PN+T\xspace}
\newcommand{\wIPNR}{\ensuremath{\omega}IPN+R\xspace}
\newcommand{\wIPNT}{\ensuremath{\omega}IPN+T\xspace}
\newcommand{\wOPNR}{\ensuremath{\omega}OPN+R\xspace}
\newcommand{\wOPNT}{\ensuremath{\omega}OPN+T\xspace}
\newcommand{\wIPN}{\ensuremath{\omega}IPN\xspace}
\newcommand{\wOPN}{\ensuremath{\omega}OPN\xspace}
\newcommand{\w}{\ensuremath{\omega}\xspace}
\newcommand{\dc}[1]{\ensuremath{\downarrow\!\!\left(#1\right)}\xspace}
\newcommand{\remIw}{\ensuremath{\textsf{remI\w}}}

\newcommand{\nbar}{\ensuremath{\overline{n}}}
\newcommand{\mbar}{\ensuremath{\overline{m}}}
\newcommand{\sigmabar}{\ensuremath{\overline{\sigma}}}

\DeclareMathOperator{\effect}{\mathit{effect}}
\DeclareMathOperator{\af}{\mathsf{AllowsFiring}}
\DeclareMathOperator{\reach}{\mathsf{Reach}}

\newcommand{\SCT}{\ensuremath{\texttt{Build-\treeshortname}}\xspace}
\newcommand{\treealgo}{\SCT}
\newcommand{\Post}{\ensuremath{\mathtt{Post}}}

\newcommand{\tuple}[1]{\langle#1\rangle\xspace}
\newcommand{\Nn}{\ensuremath{\mathcal{N}}\xspace}
\newcommand{\Tt}{\ensuremath{\mathcal{T}}\xspace}
\newcommand{\NN}{\ensuremath{\mathbb{N}}\xspace}
\newcommand{\ZZ}{\ensuremath{\mathbb{Z}}\xspace}

\newcommand{\nbw}[1]{\ensuremath{\mathrm{nb}\omega\left(#1\right)}}
\newcommand{\mgeqp}[1]{\ensuremath{\mgeq_{#1}}}
\newcommand{\mgeq}{\ensuremath{\succeq}}
\newcommand{\mleq}{\ensuremath{\preceq}}
\newcommand{\ml}{\ensuremath{\prec}}
\newcommand{\mleqp}[1]{\ensuremath{\mleq_{#1}}}
\newcommand{\0}{\ensuremath{\mathbf{0}}}

\newcommand{\expsp}{\textsc{Expspace}}

\newcommand{\NNP}{\NN^{+}}

\newcommand{\Oh}{\mathcal{O}}

\newcommand{\threshold}{h}
\newcommand{\tstep}[2]{\xrightarrow{#1}_{#2}}
\newcommand{\step}[1]{\xrightarrow{#1}}
\newcommand{\seq}{\sigma}
\newcommand{\seqbis}{\pi}
\newcommand{\maxred}{R}
\newcommand{\eff}{\mathit{effect}}
\newcommand{\lenfn}{\ell}
\newcommand{\ceilthr}[2]{[#1]_{#2 \to \omega}}
\newcommand{\floorthr}[2]{[#1]_{\omega \to #2}}
\newcommand{\cbasis}[1]{\omega(#1)}
\newcommand{\nbwb}[1]{\ensuremath{\mathrm{nb}\overline{\omega}\left(#1\right)}}
\providecommand{\infnorm}[1]{\lVert#1\rVert_{\infty}}
\newcommand{\effabs}[2]{\Delta_{#1}[#2]}
\newcommand{\numplaces}{r_{1}}
\newcommand{\effs}{\vec{B}}
\newcommand{\effel}{\vec{b}}
\newcommand{\lvv}{\vec{x}}
\newcommand{\lvvbis}{\vec{y}}
\newcommand{\fev}{\vec{d}}

\newcommand{\pom}[1]{\cbasis{#1}}

\def\ignore#1{}

\title{$\omega$-Petri nets}
\author{G. Geeraerts${}^1$ \and A. Heußner${}^2$\and M. Praveen${}^3$\and J.-F. Raskin${}^1$}
\date{
  \small
  ${}^1$ Universit\'e Libre de Bruxelles (ULB), Belgium\\
  \small
  ${}^2$ Otto-Friedrich Universit\"at Bamberg, Germany\\
  \small
  ${}^3$Laboratoire Sp\'ecification et V\'erification, ENS Cachan, France
}

\begin{document}
\maketitle

\begin{abstract}
  We introduce \w-Petri nets (\wPN), an extension of {\em plain} Petri
  nets with $\omega$-labeled input and output arcs, that is
  well-suited to analyse \emph{parametric concurrent systems with
    dynamic thread creation}. Most techniques (such as the Karp and
  Miller tree or the Rackoff technique) that have been proposed in the
  setting of \emph{plain Petri nets} do not apply directly to \wPN
  because \wPN define transition systems that have \emph{infinite
    branching}.  This motivates a thorough analysis of the
  computational aspects of \wPN. We show that an \wPN can be turned
  into an plain Petri net that allows to recover the reachability set
  of the \wPN, but that does not preserve termination. This yields
  complexity bounds for the reachability, (place) boundedness and
  coverability problems on \wPN. We provide a practical algorithm to
  compute a coverability set of the \wPN and to decide termination by
  adapting the classical Karp and Miller tree construction. We also
  adapt the Rackoff technique to \wPN, to obtain the exact complexity
  of the termination problem. Finally, we consider the extension of
  \wPN with reset and transfer arcs, and show how this extension
  impacts the decidability and complexity of the aforementioned
  problems.
\end{abstract}



\section{Introduction}

In this paper, we introduce $\omega$-Petri nets (\wPN), an extension
of {\em plain} Petri nets (PN) that allows input and output arcs to be
labeled by the symbol $\omega$, instead of a natural number.  An
$\omega$-labeled input arc consumes, non-deterministically, any number
of tokens in its input place while an $\omega$-labeled output arc
produces non-deterministically any number of tokens in its output
place.  We claim that \wPN are particularly well suited for modeling
\emph{parametric concurrent systems} (see for instance our recent work
on the Grand Central Dispatch technology \cite{GHR12}), and to perform
\emph{parametric verification} \cite{german92reasoning} on those
systems, as we illustrate now by means of the example in
Fig~\ref{fig:example-intro}. The example present a skeleton of a
distributed program, in which a \texttt{main} function forks $P$
parallel threads (where $P$ is a parameter of the program), each
executing the \texttt{one\_task} function. Many distributed programs
follow this abstract skeleton that allows to perform calculations in
parallel, and being able to model precisely such concurrent behaviors
is an important issue. In particular, we would like that the model
captures the fact that \emph{$P$ is a parameter}, so that we can, for
instance, check that the execution of the program always terminates
(assuming each individual execution of \texttt{one\_task} does),
\emph{for all possible values of~$P$}. Clearly, the Petri net (a) in
Fig.~\ref{fig:example-intro} does not capture the parametric nature of
the example, as place $p_1$ contains a fixed number $K$ of tokens. The
PN (b), on the other hand captures the fact that the program can
\texttt{fork} an unbounded number of threads, but \emph{does not
  preserve termination}: $(\mathtt{post})^\w$ is an infinite execution
of PN (b), while the programme terminates (assuming each
\texttt{one\_task} thread terminates) for all values of $P$, because
the \texttt{for} loop in line~\ref{lst:for-example-intro} executes
exactly $P$ times. Finally, observe that the \wPN (c) has the desired
properties: firing transition \texttt{fork} creates
\emph{non-deterministically} an \emph{unbounded} albeit \emph{finite}
number of tokens in $p_2$ (to model all the possible executions of the
for loop in line~\ref{lst:for-example-intro}), and all possible
executions of this \wPN terminate, because the number of tokens
produced in $p_2$ remains \emph{finite} and no further token creation
in $p_2$ is allowed after the firing of the \texttt{fork} transition.
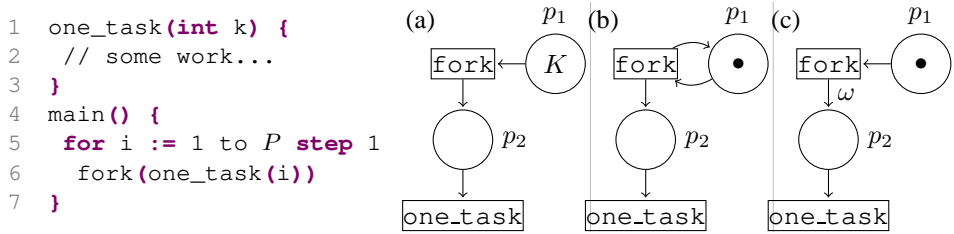
\begin{figure}[t!]
  \centering
  \begin{minipage}{.375\linewidth}
    \begin{pseudo}
one_task(int k) { 
 // some work...  
}

main() { 
 for i := 1 to $P$ step 1 @\label{lst:for-example-intro}@
  fork(one_task(i)) 
}
    \end{pseudo}
  \end{minipage}
  \begin{minipage}{.615\linewidth}
   \begin{tikzpicture}[node distance=.4cm,scale=.9]
    \begin{scope}
      \node[transition] (post) {{\tt fork}};
      \node[place, label={$p_1$}] (p1) [right=of post] {$K$};
      \node[place, label= right:{$p_2$}] (posted) [below= of post]  {};
      \node[transition ] (ot) [below=of posted] {{\tt one\_task}};
      \node[above left of= post, node distance=.8cm] {(a)} ;
      
       \path
       (post) edge [pre] (p1)
              edge [post] (posted)
       (ot)   edge [pre] (posted) ;
       
       \draw[color=lightgray] (p1.north east)+(.2,.6) -- +(.2,-2.8) ;
     \end{scope}
     \begin{scope}[shift={(2.7,0)}]
       \node[transition] (post) {{\tt fork}};
       \node[place, label={$p_1$}] (p1) [right=of post] {$\bullet$};
       \node[place, label= right:{$p_2$}] (posted) [below= of post]  {};
       \node[transition ] (ot) [below=of posted] {{\tt one\_task}};
       \node[above left of= post, node distance=.8cm] {(b)} ;

       \path
       (post) edge [pre, bend right] (p1)
              edge [post, bend left] (p1)
              edge [post] (posted)
       (ot)   edge [pre] (posted) ;       
       
       \draw[color=lightgray] (p1.north east)+(.2,.6) -- +(.2,-2.8) ;

     \end{scope} 
     \begin{scope}[shift={(5.4,0)}]
       \node[transition] (post) {{\tt fork}};
       \node[place, label={$p_1$}] (p1) [right=of post] {$\bullet$};
       \node[place, label= right:{$p_2$}] (posted) [below= of post]  {};
       \node[transition ] (ot) [below=of posted] {{\tt one\_task}};
       \node[above left of= post, node distance=.8cm] {(c)} ;

       \path
       (post) edge [pre] (p1)
              edge [post] node [right] {$\w$} (posted)
       (ot)   edge [pre] (posted) ;
     \end{scope}
  \end{tikzpicture}
\end{minipage}
  \caption{An example of a parametric system with three possible models}
  \label{fig:example-intro}
\end{figure}


While close to Petri nets, \wPN are sufficiently different that a
thorough and careful study of their computational properties is
required.  This is the main contribution of the paper.  A first
example of discrepancy is that the semantics of \wPN is an infinite
transition system which is {\em infinitely branching}. This is not the
case for plain PN: their transition systems can be infinite but they
are finitely branching. As a consequence, some of the classical
techniques for the analysis of Petri nets cannot be applied. Consider
for example the {\em finite unfolding of the transition
  system}~\cite{FS01} that stops the development of a branch of the
reachability tree whenever a node with a smaller ancestor is
found. This tree is finite (and effectively constructible) for any
plain Petri net and any initial marking because the set of markings
$\mathbb{N}^k$ is {\em well-quasi ordered}, and {\em finite branching}
of plain Petri nets allows for the use of K\"onig's lemma\footnote{In
  fact, this construction is applicable to any well-structured
  transition system which is finitely branching and allows to decide
  the termination problem for example.}. However, this technique
cannot be applied to \wPN, as they are infinitely branching. Such
peculiarities of \wPN motivate our study of three different tools for
analysing them.
First, we consider, in Section~\ref{sec:solv-term}, a variant of the
Karp and Miller tree \cite{KM69} that applies to \wPN. In order to
cope with the infinite branching of the semantics of \wPN, we need to
introduce in the Karp and Miller tree $\omega$'s that are not the
result of accelerations but the result of $\omega$-output arcs.  Our
variant of the Karp and Miller construction is {\em recursive}, this
allows us to tame the technicality of the proof, and as a consequence,
our proof when applied to {\em plain} Petri nets, provides a
simplification of the original proof by Karp and Miller.
Second, in Section~\ref{sec:solv-cover-place}, we show how to
construct, from an \wPN, a plain Petri net that preserve its
reachability set. This reduction allows to prove that many bounds on
the algorithmic complexity of (plain) PN problems apply to \wPN
too. However, it does not preserve {\em termination}. Thus, we study,
in Section~\ref{sec:extend-rack-techn}, as a third contribution, an
extension of the self-covering path technique due to
Rackoff~\cite{Rackoff78}. This technique allows to provide a direct
proof of {\sc ExpSpace} upper bounds for several classical decision
problems, and in particular, this allows to prove {\sc ExpSpace}
completeness of the termination problem.

Finally, in Section~\ref{sec:extensions}, as a additional
contribution, and to get a complete picture, we consider extensions of
\wPN with {\em reset} and {\em transfer} arcs~\cite{DFS98}. For those
extensions, the decidability results for reset and transfer nets
(without $\omega$ arcs) also apply to our extension with the notable
exception of the termination problem that becomes, as we show here,
undecidable.  The summary of our results are given in
Table~\ref{tab:complexity}.

\begin{table}[t!]
  \caption{Complexity results on \wPN (with the section numbers   
    where the results are proved). \wIPNR (\wOPNR) and \wIPNT 
    (\wOPNT) denote resp. Petri nets with reset (R) and transfer 
    (T) arcs with \w on input (output) arcs only.}
  \centering
  \begin{tabular}{|l||p{.23\linewidth}|p{.23\linewidth}|p{.23\linewidth}|}
    \hline
    Problem           &\wPN                                                          &\wPNT                                                 &\wPNR\\
    \hline\hline
    Reachability      
    &Decidable and \textsc{ExpSpace}-hard (\ref{sec:reach})
    &\multirow{2}{\linewidth}{Undecidable (\ref{sec:extensions})}             
    &\multirow{3}{\linewidth}{Undecidable (\ref{sec:extensions})} \\\cline{1-2}
    %
    %
    Place-boundedness 
    &\multirow{3}{\linewidth}{\textsc{ExpSpace}-c (\ref{sec:reach})} 
    & 
    & 
    \\\cline{1-1}\cline{3-3}
    Boundedness
    & 
    &Decidable (\ref{sec:extensions})             
    & 
    \\\cline{1-1}\cline{3-4}
    Coverability 
    &  
    &\multicolumn{2}{|c|}{Decidable and Ackerman-hard (\ref{sec:extensions})}\\\hline
  \end{tabular}
  \medskip

  \begin{tabular}{|l||p{.25\linewidth}|p{.27\linewidth}|p{.25\linewidth}|}
    \hline
    Problem        &\wPN                                                  &\wOPNT, \wOPNR                        &\wIPNT, \wIPNR\\
    \hline\hline
    Termination    &\textsc{ExpSpace}-c (\ref{sec:extend-rack-techn})     &Undecidable (\ref{sec:extensions})    &Decidable and Ackerman-hard (\ref{sec:extensions})    \\
    \hline
  \end{tabular}
  \label{tab:complexity}
\end{table}

\paragraph{{\bf Related works}}
\wPN are well-structured transition systems~\cite{FS01}. The set
saturation technique~\cite{ACJT96} and so symbolic backward analysis
can be applied to them while the finite tree unfolding is not
applicable because of the infinite branching property of \wPN. For the
same reason, \wPN are \emph{not} well-structured nets \cite{FMP04}.

In~\cite{BJK2010}, Bradzil \textit{et al.} extends the Rackoff
technique to VASS games with $\omega$ output arcs. While this
extension of the Rackoff technique is technically close to ours, we
cannot directly use their results to solve the termination problem of
\wPN.

Several works (see for instance \cite{DRV02,Del03} rely on Petri nets
to model \emph{parametric systems} and perform \emph{parametrised
  verification}. However, in all these works, the dynamic creation of
threads uses the same pattern as in Fig.~\ref{fig:example-intro} (b),
and does not preserve termination. \wPN allow to model more faithfully
the dynamic creation of an unbounded number of threads, and are thus
better suited to model new programming paradigms (such as those use in
GCD \cite{GHR12}) that have been recently proposed to better support
multi-core platforms.

{\bf Remark}: due to lack of space, most proofs can be found in the
appendix.


\section{\w-Petri nets\label{sec:omega-petri-nets}}
Let us define the syntax and semantics of our Petri net extension,
called \emph{$\omega$ Petri nets} (\wPN for short). Let $\omega$ be a
symbol that denotes `any positive integer value'. We extend the
arithmetic and the $\leq$ ordering on $\ZZ$ to $\ZZ\cup\{\w\}$ as
follows: $\w+\w=\w-\w=\w$; and for all $c\in\ZZ$: $c+\w=\w+c=\w-c=\w$;
$c-\w=c$; and $c\leq\w$. The fact that $c-\w=c$ might sound surprising
but will be justified later when we introduce $\wPN$. An
\emph{\w-multiset} (or simply \emph{multiset}) of elements from $S$ is
a function $m:S\mapsto\mathbb{N}\cup\{\w\}$. We denote multisets $m$
of $S=\{s_1, s_2,\ldots, s_n\}$ by extension using the syntax
$\{m(s_1)\otimes s_1, m(s_2)\otimes s_2,\ldots, m(s_n)\otimes s_n\}$
(when $m(s)=1$, we write $s$ instead of $m(s)\otimes s$, and we omit
elements $m(s)\otimes s$ when $m(s)=0$). Given two multisets $m_1$ and
$m_2$, and an integer value $c$ we let $m_1+ m_2$ be the multiset
s.t. $(m_1+ m_2)(p)=m_1(p)+m_2(p)$; $m_1-m_2$ be the multiset
s.t. $(m_1- m_2)(p)=m_1(p)-m_2(p)$; and $c\cdot m_1$ be the multiset
s.t. $(c\cdot m_1)(p)=c\times m_1(p)$ for all $p\in P$.

\paragraph{Syntax} Syntactically, \wPN extend plain Petri nets
\cite{petri62institut,reisig} by allowing (input and output) arcs to
be labeled by $\w$. Intuitively, if a transition $t$ has $\w$ as
output (resp. input) effect on place $p$, the firing of $t$
non-deterministically creates (consumes) a positive number of tokens
in $p$.

\begin{definition}
  A \emph{Petri net with \w-arcs} (\wPN) is a tuple $\Nn=\tuple{P,T}$
  where: $P$ is a finite set of \emph{places}; $T$ a finite set of
  \emph{transitions}. Each transition is a pair $t=(I,O)$, where:
  $I:P\rightarrow \NN\cup\{\w\}$ and $O:P\rightarrow \NN\cup\{\w\}$,
  give respectively the input (output) effect $I(p)$ ($O(p)$) of $t$
  on place $p$.
\end{definition}
By abuse of notation, we denote by $I(t)$ (resp. $O(t)$) the functions
s.t. $t=(I(t), O(t))$. When convenient, we sometimes regard $I(t)$ or
$O(t)$ as \emph{\w-multisets} of places.  Whenever there is $p$
s.t. $O(t)(p)=\w$ (resp. $I(t)(p)=\w$), we say that $t$ is an
\emph{$\w$-output-transition} (\emph{$\w$-input-transition}). A
transition $t$ is an $\w$-transition iff it is an
$\w$-output-transition or an $\w$-input-transition. Otherwise, $t$ is
a \emph{plain} transition. Remark that a (plain) Petri net is an \wPN
with plain transitions only. Moreover, when an \wPN contains no
$\w$-output-transitions (resp. no $\w$-input transitions), we say that
it is an $\w$-input-PN ($\w$-output-PN), or \wIPN (\wOPN) for
short. For all transitions $t$, we denote by $\effect(t)$ the function
$O(t)-I(t)$. Remark that $\effect(t)(p)$ could be $\omega$ for some
$p$ (in particular when $O(t)(p)=I(t)(p)=\omega$).  Intuitively,
$\effect(t)(p)=\omega$ models the fact that firing $t$ can increase
the marking of $p$ by an arbitrary number of tokens. Finally, observe
that $O(t)(p)=c\neq\w$ and $I(t)(p)=\w$ \emph{implies}
$\effect(t)(p)=c-\w=c$. This models the fact that firing $t$ can at
most increase the marking of $p$ by $c$ tokens. Thus, intuitively, the
value $\effect(t)(p)$ models the \emph{maximal possible effect} of $t$
on $p$. We extend the definition of $\effect$ to sequences of
transitions $\sigma=t_1t_2\cdots t_n$ by letting
$\effect(\sigma)=\sum_{i=1}^n \effect(t_i)$.

A \emph{marking} is a function $P\mapsto\NN$. An \emph{$\w$-marking}
is a function $P\mapsto\NN\cup\{\w\}$, i.e. an \w-multiset
on~$P$. Remark that any marking is an $\w$-marking, and that, for all
transitions $t=(I,O)$, $I$ and $O$ are both $\w$-markings.  We denote
by ${\bf 0}$ the marking s.t. ${\bf 0}(p)=0$ for all $p\in P$. For all
\w-markings $m$, we let $\pom{m}$ be the set of places $\{p\mid
m(p)=\omega\}$, and let $\nbw{m}=|\pom{m}|$.  We define \emph{the
  concretisation} of $m$ as the set of all markings that coincide with
$m$ on all places $p\not\in\pom{m}$, and take an arbitrary value in
any place from $\pom{m}$. Formally: $\gamma(m)=\{m'\mid \forall
p\not\in \pom{m}: m'(p)=m(p)\}$. We further define a family of
orderings on $\w$-markings as follows. For any $P'\subseteq P$, we let
$m_1\mleqp{P'}m_2$ iff $(i)$ for all $p\in P'$: $m_1(p)\leq m_2(p)$,
and $(ii)$ for all $p\in P\setminus P'$: $m_1(p)=m_2(p)$. We
abbreviate $\mleqp{P}$ by $\mleq$ (where $P$ is the set of places of
the \wPN). It is well-known that $\mleq$ is a \emph{well-quasi
  ordering} (wqo), that is, we can extract, from any infinite sequence
$m_1, m_2,\ldots, m_i,\ldots$ of markings, an infinite subsequence
$\mbar_1, \mbar_2,\ldots, \mbar_i,\ldots$ s.t. $\mbar_i\mleq
\mbar_{i+1}$ for all $i\geq 1$. For all \w-markings $m$, we let
$\dc{m}$ be the \emph{downward-closure} of $m$, defined as
$\dc{m}=\{m'\mid m'\textrm{ is a marking and } m'\mleq m\}$. We extend
$\downarrow$ to sets of \w-markings: $\dc{S}=\cup_{m\in S}\dc{m}$. A
set $D$ of markings is \emph{downward-closed} iff $\dc{D}=D$. It is
well-known that (possibly infinite) downward-closed sets of markings
can always be represented by a finite set of \w-markings, because the
set of \w-markings forms an \emph{adequate domain of limits}
\cite{DBLP:journals/jcss/GeeraertsRB06}: for all downward-closed sets
$D$ of markings, there exists a finite set $M$ of \w-markings
s.t. $\dc{M}=D$. We associate, to each \wPN, an \emph{intial marking}
$m_0$. From now on, we consider mostly initialised \wPN $\tuple{P, T,
  m_0}$.

\begin{figure}
  \centering
  \begin{tikzpicture}
    \node[place,label={$p_1$}] (p1) {$\bullet$} ;
    \node[transition,label={$t_1$}] (t1) [right=of p1] {} ;
    \node[place,label={$p_2$}] (p2) [right=of t1] {} ;
    \node[transition, label={$t_2$}] (t2) [right=of p2] {} ;
    \node[place, label={$p_3$}] (p3) [right=of t2] {} ;
    \node[transition, label=right:{$t_4$}, color=red] (t4) [below=of p3] {} ;
    \node[transition, label=right:{$t_3$}] (t3) [below=of p2] {} ;

    \path
    (t1) edge [pre] (p1)
         edge [post] node [above] {$\w$} (p2)
    (t3) edge [pre] (p2)
    (t2) edge [post] node [above] {$2$} (p3)
         edge [pre] (p2)
    (t4) edge [post, bend right, color=red] (p3) 
         edge [pre, bend left, color=red] (p3) ;
  \end{tikzpicture}
  \caption{An example \wPN $\Nn_1$. The \wPN $\Nn_1'$ is obtained by
    removing transition $t_4$ (red).}
  \label{fig:examplewPN}
\end{figure}
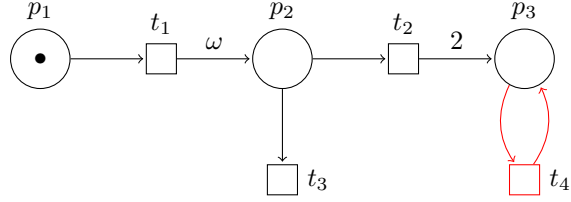

\begin{example}
  An example of an \wPN (actually an \wOPN) $\Nn_1=\tuple{P,T,m_0}$ is
  shown in Fig.~\ref{fig:examplewPN}. In this example, $P=\{p_1, p_2,
  p_3\}$, $T=\{t_1,t_2,t_3,t_4\}$, $m_0(p_1)=1$ and
  $m_0(p_2)=m_0(p_3)=0$. $t_1$ is the only $\w$-transition, with
  $O(t_1)(p_2)=\w$. This \wPN will serve as a running example
  throughout the section.
\end{example}

\paragraph{Semantics} Let $m$ be an \emph{$\w$-marking}. A transition
$t=(I,O)$ is \emph{firable from $m$} iff: $m(p) \mgeq I(p)$ for all
$p$ s.t. $I(p)\neq\omega$. We consider two kinds of possible effects
for~$t$. The first is the \emph{concrete semantics} and applies only
when $m$ is a \emph{marking}. In this case, firing $t$ yields a new
marking $m'$ s.t. for all $p\in P$: $m'(p)=m(p)-i+o$ where:
$i=I(t)(p)$ if $I(t)(p)\neq \omega$, $i\in \{0,\ldots,m(p)\}$ if
$I(t)(p)=\omega$, $o=O(t)(p)$ if $O(t)(p)\neq\omega$ and $o\geq 0$ if
$O(t)(p)=\omega$. This is denoted by $m\xrightarrow{t}m'$. Thus,
intuitively, $I(t)(p)=\omega$ (resp. $O(t)(p)=\omega$) means that $t$
consumes (produces) an arbitrary number of tokens in $p$ when
fired. Remark that, in the concrete semantics, $\w$-transitions are
\emph{non-deterministic}: when $t$ is an $\w$-transitions that is
firable in $m$, there are \emph{infinitely many} $m'$
s.t. $m\xrightarrow{t} m'$. The latter semantics is the
\emph{$\w$-semantics}. In this case, firing $t=(I,O)$ yields the
(unique) $\w$-marking $m'=m-I+O$ (denoted $m\xrightarrow{t}_\w
m'$). Remark that $m\xrightarrow{t}m'$ iff $m\xrightarrow{t}_\w m'$
when $m$ and $m'$ are markings.

We extend the $\rightarrow$ and $\rightarrow_\w$ relations to finite
or infinite sequences of transitions in the usual way. Also we write
$m\xrightarrow{\sigma}$ iff $\sigma$ is \emph{firable} from
$m$. More precisely, for a finite sequence of transitions
$\sigma=t_1\cdots t_n$, we write $m\xrightarrow{\sigma}$ iff there
are $m_1$, \ldots, $m_n$ s.t. for all $1\leq i\leq n$:
$m_{i-1}\xrightarrow{t_i}m_i$. For an infinite sequence of transitions
$\sigma=t_1\cdots t_j\cdots$, we write $m_0\xrightarrow{\sigma}$ iff
there are $m_1, \ldots, m_j,\ldots$ s.t. for all $i\geq 1$:
$m_{i-1}\xrightarrow{t_i}m_i$.

Given an \wPN $\Nn=\tuple{P,T,m_0}$, an \emph{execution} of $\Nn$ is
either a finite sequence of the form $m_0,t_1,m_1,t_2,\ldots,t_n, m_n$
s.t. $m_0\xrightarrow{t_1}m_1\xrightarrow{t_2}\cdots
\xrightarrow{t_n}m_n$, or an infinite sequence of the form
$m_0,t_1,m_1,t_2,\ldots,t_j, m_j,\ldots$ s.t. for all $j\geq 1$:
$m_{j-1}\xrightarrow{t_j} m_j$. We denote by $\reach(\Nn)$ the set of
markings $\{m\mid \exists \sigma\textrm{ s.t. }
m_0\xrightarrow{\sigma} m\}$ that are reachable from $m_0$ in
$\Nn$. Finally, a \emph{finite set of \w-markings} ${\cal CS}$ is a
\emph{coverability set} of $\Nn$ (with initial marking $m_0$) iff
$\dc{{\cal CS}}=\dc{\reach(\Nn)}$. That is, any coverability set
${\cal CS}$ is a \emph{finite representation of the downward-closure
  of $\Nn$'s reachable markings}.

\begin{example}\label{example:intro-2}
  The sequence $t_1t_2^K$ is firable for all $K\geq 0$ in $\Nn_1$
  (Fig.~\ref{fig:examplewPN}). Indeed, for each $K\geq 0$, one
  possible execution corresponding to $t_1t_2^K$ is given by
  $\tuple{1,0,0} \xrightarrow{t_1}\tuple{0,3K,0}
  \xrightarrow{t_2}\tuple{0,3K-1,2} \xrightarrow{t_2}\tuple{0,3K-2,4}
  \xrightarrow{t_2}\cdots \xrightarrow{t_2}\tuple{0,2K,2K}$. Remark
  that there are other possible executions corresponding to the same
  sequence of transitions, because the number of tokens created by
  $t_1$ in $p_2$ is chosen non-deterministically. Also,
  $t_1t_2t_4^\omega$ is an infinite firable sequence of
  transitions. Finally, observe that the set of reachable markings in
  $\Nn_1$ is $\reach(\Nn)=\{\tuple{1,0,0}\}\cup\{\tuple{0, i, 2\times
    j}\mid i,j\in\mathbb{N}\}$. The set of \w markings ${\cal
    CS}=\{\tuple{1,0,0}, \tuple{0,\w,\w}\}$ is a coverability set of
  $\Nn$. Note that $\dc{{\cal CS}}\supsetneq \reach(\Nn)$: for
  instance, $\tuple{0,1,1}\in \dc{{\cal CS}}$, but $\tuple{0,1,1}$ is
  not reachable.
\end{example}
  
Let us now observe two properties of the semantics of \wPN, that will
be useful for the proofs of Section~\ref{sec:solv-term}. The first
says that, when firing a sequence of transitions $\sigma$ that have
non \w-labeled arcs on to and from some place $p$, the effect of
$\sigma$ on $p$ is as in a plain PN:
\begin{lemma}\label{lemma:when-no-omega-the-effect-is-the-same-as-PN}
  Let $m$ and $m'$ be two markings and let $\sigma=t_1\cdots t_n$ be a
  sequence of transitions of an \wPN s.t. $m\xrightarrow{\sigma}
  m'$. Let $p$ be a place s.t. for all $1\leq i\leq n$:
  $O(t_i)(p)\neq\w\neq I(t_i)(p)$. Then,
  $m'(p)=m(p)+\effect(\sigma)(p)$.
\end{lemma}
The latter property says that the set of markings that are reachable
by a given sequence of transitions $\sigma$ is upward-closed
w.r.t. $\mleqp{P'}$, where $P'$ is the set of places where the effect
of $\sigma$ is $\w$.
\begin{lemma}\label{lemma:reachable-markings-is-upward-closed}
  Let $m_1$, $m_2$ and $m_3$ be three markings, and let $\sigma$ be a
  sequence of transitions s.t. $(i)$ $m_1\xrightarrow{\sigma}m_2$,
  $(ii)$ $m_3\mgeqp{P'} m_2$ with $P'=\{p\mid
  \effect(\sigma)(p)=\w\}$. Then, $m_1\xrightarrow{\sigma}m_3$ holds
  too.
\end{lemma}

\paragraph{Problems}
We consider the following problems. Let $\Nn=(P,T,m_0)$ be an \wPN:
\begin{enumerate}
\item The \emph{reachability problem} asks, given a marking $m$,
  whether $m\in\reach(N)$.
\item The \emph{place boundedness problem} asks, given a place $p$ of
  $\Nn$, whether there exists $K\in\mathbb{N}$ s.t. for all $m\in
  \reach(\Nn)$: $m(p)\leq K$. If the answer is positive, we say that
  $p$ is \emph{bounded} (from $m_0$).
\item The \emph{boundedness problem} asks whether all places of $\Nn$
  are bounded (from $m_0$).
\item The \emph{covering problem} asks, given a marking $m$ of $\Nn$,
  whether there exists $m'\in\reach(\Nn)$ s.t. $m'\mgeq m$.
\item The \emph{termination problem} asks whether all executions of
  $\Nn$ are finite.
\end{enumerate}

Remark that a \emph{coverability set} of the \wPN is sufficient to
solve \emph{boundedness}, \emph{place boundedness} and
\emph{covering}, as in the case of Petri nets. If ${\cal CS}$ is a
coverability set of $\Nn$, then: $(i)$ $p$ is bounded iff
$m(p)\neq\omega$ for all $m\in {\cal CS}$; $(ii)$ $\Nn$ is bounded iff
$m(p)\neq\omega$ for all $p$ and for all $m\in {\cal CS}$; and
$(iii)$, $\Nn$ can cover $m$ iff there exists $m'\in {\cal CS}$
s.t. $m\mleq m'$.  As in the plain Petri nets case, a sufficient and
necessary condition of non-termination is the existence of a
\emph{self covering execution}.  A \emph{self covering execution} of
an \wPN $\Nn=\tuple{P,T,m_0}$ is a \emph{finite} execution of the form
$m_0\xrightarrow{t_1}m_1\cdots
\xrightarrow{t_k}m_k\xrightarrow{t_{k+1}}\cdots\xrightarrow{t_n}m_n$
with $m_n\mgeq m_k$:

\begin{lemma}\label{lemma:terminates-iff-no-s-c-e}
  An \wPN terminates iff it admits no self-covering execution.
\end{lemma}

\begin{example}
  Consider again the \wPN $\Nn_1$ in Fig.~\ref{fig:examplewPN}. Recall
  from Example~\ref{example:intro-2} that, for all $K\geq 0$,
  $t_1t_2^K$ is firable and allows to \emph{reach}
  $\tuple{0,2K,2K}$. All these markings are thus \emph{reachable}.
  These sequences of transitions also show that $p_2$ and $p_3$ are
  \emph{unbounded} (hence, $\Nn_1$ is unbounded too), while $p_1$ is
  \emph{bounded}. Marking $\tuple{0,1,1}$ is \emph{not reachable} but
  \emph{coverable}, while $\tuple{2,0,0}$ is neither reachable nor
  coverable. Finally, $\Nn_1$ does not terminate (because
  $t_1t_2t_4^\omega$ is firable), while $\Nn_1'$ does. In particular,
  \emph{in $\Nn_1'$}, $t_3$ can fire only a \emph{finite} number of
  time, because $t_1$ will always create a finite (albeit unbounded)
  number of tokens in $p_2$. This an important difference between \wPN
  and plain PN: no unbounded PNs terminates, while there are unbounded
  \wPN that terminate, e.g. $\Nn_1'$.
\end{example}


\section{A Karp and Miller procedure for
  \wPN\label{sec:solv-term}}
In this section, we presents an extension of the classical \treename
procedure \cite{KM69}, adapted to \wPN. We show that the finite tree
built by this algorithm (coined the \treeshortname tree), allows, as
in the case of PNs, to decide \emph{boundedness}, \emph{place
  boundednes}, \emph{coverability} and \emph{termination} on \wPN.

Before describing the algorithm, we discuss intuitively the
\treeshortname trees of the \wPN $\Nn_1$ and $\Nn_1'$ given in
Fig.~\ref{fig:examplewPN}. Their respective \treeshortname trees (for
the initial marking $m_0=\tuple{1,0,0}$) are $\Tt_1$ and $\Tt_1'$,
respectively the tree in Fig.~\ref{fig:example-tree} and its
\emph{black subtree} (i.e., excluding $n_7$). As can be observed, the
nodes and edges of a \treeshortname tree are labeled by \w-markings
and transitions respectively. The relationship between a
\treeshortname tree and the executions of the corresponding \wPN can
be formalised using the notion of \emph{stuttering path}. Intuitively,
a stuttering path is a sequence of nodes $n_1,n_2,\ldots, n_k$
s.t. for all $i\geq 2$: either $n_i$ is a son of $n_{i-1}$, or $n_i$
is an \emph{ancestor} of $n_{i-1}$ \emph{that has the same label} as
$n_{i-1}$.  For instance, $\pi=n_1,n_2,n_4,n_2,n_3,n_6, n_3, n_5, n_3,
n_5$ is a stuttering path in $\Tt_1'$. Then, we claim $(i)$ that
\emph{every execution of the \wPN is simulated by a stuttering path}
in its \treeshortname tree, and that $(ii)$ \emph{every stuttering
  path in the \treeshortname tree corresponds to a family of
  executions of the \wPN}, where an arbitrary number of tokens can be
produced in the places marked by \w in the \treeshortname tree. For
instance, the execution $m_0,t_1,\tuple{0, 42, 0},t_3,\tuple{0, 41,
  0},t_2,\tuple{0, 40, 2},t_3,\tuple{0, 39, 2},t_2,\tuple{0, 38,
  4},t_2,\tuple{0, 37, 6}$, of $\Nn_1'$ is witnessed in $\Tt_1'$ by
the stuttering path $\pi$ given above -- observe that the sequence of
edge labels in $\pi$'s equals the sequence of transitions of the
execution, and that all markings along the execution are
\emph{covered} by the labels of the corresponding nodes in $\pi$:
$m_0\in\gamma(n_1)$, $\tuple{0, 42, 0}\in\gamma(n_2)$, and so
forth. On the other hand, the stuttering path $n_1,n_2,n_3$ of $\Nn_1$
summarises all the (infinitely many) possible executions obtained by
firing a sequence of the form $t_1t_2^n$. Indeed, for all $k\geq 1$,
$\ell\geq 0$: $m_0,t_1,\tuple{0,k+\ell,0},t_2,\tuple{0,
  k+\ell-1,2},t_2,\ldots,t_2,\tuple{0,k,2\times\ell}$ is an execution of
$\Nn_1$, so, an arbitrary number of tokens can be obtained in both
$p_2$ and $p_3$ by firing sequences of the form $t_1t_2^n$. Finally,
observe that a \emph{self-covering execution} of $\Nn_1$, such as
$m_0,t_1,\tuple{0,1,0},t_2, \tuple{0,0,2},t_4,\tuple{0,0,2}$ can be
detected in $\Tt_1$, by considering the path $n_1,n_2,n_3,n_7$, and
noting that the label of $(n_3,n_7)$ is $t_4$ with
$\effect(t_4)\mgeq{\bf 0}$.

 \begin{figure}[t!]
   \centering
   \begin{tikzpicture}
     \node[draw, label=above right:{$n_1$}] (n1) {$\tuple{1,0,0}$} ;
     \node[draw, label=above right:{$n_2$}] (n2) [below=of n1] {$\tuple{0,\w,0}$} ;
     \node[draw, label=above left:{$n_3$}] (n3) [below left=of n2] {$\tuple{0,\w,\w}$} ;
     \node[draw, label=above right:{$n_4$}] (n4) [below right=of n2] {$\tuple{0,\w,0}$} ;
     \node[draw, label=above left:{$n_5$}] (n5) [below left=of n3] {$\tuple{0,\w,\w}$} ;
     \node[draw, label=above right:{$n_6$}] (n6) [below=of n3] {$\tuple{0,\w,\w}$} ;
     \node[draw, label=above right:{\color{red}$n_7$}, color=red] (n7) [below right=of n3] {$\tuple{0,\w,\w}$} ;
     
     \path[-latex]
     (n1) edge node [left] {$t_1$} (n2)
     (n2) edge node [above] {$t_2$} (n3)
          edge node [above] {$t_3$} (n4)
     (n3) edge node [above] {$t_2$} (n5)
          edge[color=red] node [above, color=red] {$t_4$} (n7)
          edge node [left] {$t_3$} (n6)
     ;  
   \end{tikzpicture}
   \caption{The \treeshortname trees $\Tt_1$ (whole tree) and $\Tt_1'$
     (black subtree) of resp. $\Nn_1$ and
     $\Nn_1'$.\label{fig:example-tree}}

\end{figure}

\paragraph{The \treealgo algorithm} Let us now show how to build
algorithmically the \treeshortname of an \wPN. Recall that, \emph{in
  the case of plain PNs}, the Karp\& Miller tree \cite{KM69} can be
regarded as a \emph{finite over-approximation of the (potentially
  infinite) reachability tree of the PN}. Thus, the Karp\& Miller
algorithm works by unfolding the transition relation of the PN, and
adds two ingredients to guarantee that the tree is finite. First, a
node $n$ that has an ancestor $n'$ \emph{with the same label} is
\emph{not developed} (it has no children). Second, when a node $n$
with label $m$ has an ancestor $n'$ with label $m'\ml m$, an
\emph{acceleration function} is applied to produce a marking $m_\w$
s.t. $m_\w(p)=\w$ if $m(p)>m'(p)$ and $m_\w(p)=m(p)$ otherwise. This
acceleration is \emph{sound} wrt to coverability since the sequence of
transition that has produced the branch $(n,n')$ can be iterated an
arbitrary number of times, thus producing arbitrary large numbers of
tokens in the places marked by $\w$ in $m_\w$. Remark that these two
constructions are not sufficient to ensure termination of the
algorithm in the case of \wPN, as \wPN are \emph{not finitely
  branching} (firing an $\w$-output-transition can produce infinitely
many different successors). To cope with this difficulty, our solution
unfolds the \emph{$\w$-semantics} $\rightarrow_\w$ instead of the
concrete semantics $\rightarrow$. This has an important consequence:
whereas the presence of a node labeled by $m$ with $m(p)=\w$ in the
\treeshortname tree of a PN $\Nn$ \emph{implies} that $\Nn$ \emph{does
  not terminate}, this is \emph{not true anymore} in the case of
\wPN. For instance, all nodes but $n_1$ in $\Tt_1'$
(Fig.~\ref{fig:example-tree}) are marked by $\w$, yet the
corresponding \wPN $\Nn_1'$ (Fig.~\ref{fig:examplewPN}) \emph{does
  terminate}.

Our version of the \treename tree adapted to \wPN is given in
Fig.~\ref{fig:scg}. It builds a tree $\Tt=\tuple{N,E,\lambda,\mu,n_0}$
where: $N$ is a set of nodes; $E\subseteq N\times N$ is a set of
edges; $\lambda:N\mapsto (\NN\cup\{\w\})^P$ is a function that labels
nodes by $\omega$-markings\footnote{We extend $\lambda$ to set of
  nodes $S$ in the usual way: $\lambda(S)=\{\lambda(n)\mid n\in
  S\}$.}; $\mu:E\mapsto T$ is a labeling function that labels arcs by
transitions; and $n_0\in N$ is the root of the tree.
For each edge $e$, we let $\effect(e)=\effect(\mu(e))$.  Let $E^+$ and
$E^*$ be respectively the transitive and the transitive reflexive
closure of $E$.  A \emph{stuttering path} is a finite sequence
$n_0,n_1,\ldots, n_\ell$ s.t. for all $1\leq i\leq \ell$:
\emph{either} $(n_{i-1},n_i)\in E$ \emph{or} $(n_i, n_{i-1})\in E^+$
and $\lambda(n_i)=\lambda(n_{i-1})$. A stuttering path
$n_0,n_1,\ldots, n_\ell$ is a \emph{(plain) path} iff
$(n_{i-1},n_i)\in E$ for all $1\leq i\leq \ell$.  Given two nodes $n$
and $n'$ s.t. $(n,n')\in E^*$, we denote by $n\leadsto n'$ the (unique
path) from $n$ to $n'$. Given a stuttering path $\pi=n_0,n_1,\ldots,
n_\ell$, we denote by $\mu(\pi)$ the sequence
$\mu(n_0,n_1)\mu(n_1,n_2)\cdots\mu(n_{\ell-1},n_\ell)$ assuming
$\mu(n_i, n_{i+1})=\varepsilon$ when $(n_i,n_{i+1})\not \in E$; and by
$\effect(\pi)=\sum_{i=1}^\ell\effect(n_{i-1},n_i)$, letting
$\effect(n_{i-1}, n_i)={\bf 0}$ when $(n_i,n_{i+1})\not \in E$.

\begin{figure}[t!]
\begin{description}
  \item[Input] an \wOPN $\Nn=\tuple{P,T}$ and an $\omega$-marking $m_0$
  \item[Output] the \treeshortname of $\Nn$, starting from $m_0$
\end{description}
{\treealgo($\Nn,m_0$):}
\begin{pseudo}[firstnumber=1]
$\Tt$ := $\tuple{N,E,\lambda,\mu,n_0}$ where $N=\{n_0\}$ with $\lambda(n_0)=m_0$
U := $\{n_0\}$

while U $\neq\emptyset$: @\label{forall:unkown}@
   select and remove $n$ from U @\label{lst:select}@
   if $\nexists \overline{n}$ st $(\overline{n},n)\in E^+$ and $\lambda(n)=\lambda(\overline{n})$: @\label{lst:if2}@
      forall $t$ in $T$ s.t. $\forall p\in P$: $I(t)(p)\neq\w$ implies $\lambda(n)(p)\geq I(t)(p)$: @\label{lst:forall_t}@
        $m'$ := $\Post$($\Nn$,$\lambda(n)$, $t$) @\label{lst:call_to_Post}@
        if $\nbw{m'} > \nbw{\lambda(n)}$:
          $\Tt'$ := $\SCT$($\Nn$,$m'$) @\label{lst:recursive_call}@
          add all edge and nodes of $\Tt'$ to $\Tt$
          let $n'$ be the root of $\Tt'$
        else
           $n'$ := new node with $\lambda(n')=m'$
           U := U $\cup$ $\{n'\}$
        E := $\mathtt{E} \cup (n,n')$ s.t. $\mu(n,n')=t$. @\label{lst:connect}@
return $\Tt$ @\label{lst:return}@
\end{pseudo}
\vspace{3ex}

\Post($\Nn$,$n$,$t$):\lstset{firstnumber=last}
\begin{pseudo}
$m'$ := $\lambda(n)-I(t)+O(t)$ @\label{lst:effect-trans}@
if  $\exists \overline{n}: \left( \overline{n},n)\in E^+\wedge \lambda(\overline{n})\ml\lambda(n)\right)$: @\label{lst:if-post}@
  $m_w(p):=\begin{cases}m'(p) & if \effect(\overline{n}\leadsto n\cdot t)(p)\leq 0\\\omega & otherwise\end{cases}$ @\label{lst:def-m-w}@
  return $m_w$
else:
  return $m'$
\end{pseudo}

\caption{The algorithm to build the \treeshortname of an
  \wPN.\label{fig:scg}}
\end{figure}

\treealgo follows the intuition given above. At all times, it
maintains a frontier $\mathtt{U}$ of tree nodes that are candidate for
development (initially, $\mathtt{U}=\{n_0\}$, with
$\lambda(n_0)=m_0$). Then, \treealgo iteratively picks up a node $n$
from $\mathtt{U}$ (see line~\ref{lst:select}), and develops it
(line~\ref{lst:forall_t} onwards) if $n$ has no ancestor $n'$ with the
same label (line~\ref{lst:if2}). \emph{Developing} a node $n$ amounts
to computing all the marking $m$ s.t. $\lambda(n)\rightarrow_\w m$
(line~\ref{lst:effect-trans}), performing accelerations
(line~\ref{lst:def-m-w}) if need be, and inserting the resulting
children in the tree. Remark that \treealgo is \emph{recursive} (see
line~\ref{lst:recursive_call}): every time a marking $m$ with an extra
$\w$ is created, it performs a recursive call to $\treealgo(\Nn, m)$,
using $m$ as initial marking\footnote{Although this differs from
  classical presentations of the Karp\& Miller technique, we have
  retained it because it simplifies the proofs of correctness.}.

The rest of the section is devoted to proving that this algorithm is
correct. We start by establishing termination, then soundness (every
stuttering path in the tree corresponds to an execution of the \wOPN)
and finally completeness (every execution of the \wOPN corresponds to
a stuttering path in the tree). To this end, we rely on the following
notions. Symmetrically to \emph{self-covering executions} we define
the notion of \emph{self-covering (stuttering) path} in a tree: a
(stuttering) path $\pi$ is \emph{self-covering} iff $\pi=\pi_1\pi_2$
with $\effect(\pi_2)\geq \0$. A self-covering stuttering path
$\pi=\pi_1\pi_2$ is \emph{$\omega$-maximal} iff for all nodes $n$,
$n'$ along $\pi_2$: $\nbw{n}=\nbw{n'}$.

\paragraph{Termination} Let us show that \treealgo always
terminates. First observe that the depth of recursive calls is at most
by $|P|+1$, as the number of places marked by $\w$ along a branch does
not decrease, and since we perform a recursive call only when a place
gets marked by $\w$ and was not before. Moreover, the branching degree
of the tree is bounded by the number $|T|$ of transitions. Thus, by
K\"onig's lemma, an infinite tree would contain an infinite branch. We
rule out this possibility by a classical wqo argument: if there were
an infinite branch in the tree computed by $\SCT(\Nn,m_0)$, then there
would be two nodes $n_1$ along the branch $n_2$ (where $n_1$ is an
ancestor of $n_2$) s.t. $\lambda(n_1)\mleq\lambda(n_2)$ and
$\effect(n_1\leadsto n_2)\mgeq 0$. Since the depth of recursive calls
is bounded, we can assume, wlog, that $n_1$ and $n_2$ have been built
during the same recursive call, hence $\lambda(n_1)\ml\lambda(n_2)$ is
not possible, because this would trigger an acceleration, create an
extra \w and start a new recursive call. Thus,
$\lambda(n_1)=\lambda(n_2)$, but in this case the algorithm stops
developing the branch (line \ref{lst:if2}). See the appendix for a
full proof.
\begin{proposition}\label{prop:termination}
  For all \wPN $\Nn$ and for all marking $m_0$, $\SCT(\Nn,m_0)$
  terminates.
\end{proposition}

Then, following the intuition that we have sketched at the beginning
of the section, we show that \treeshortname is \emph{sound}
(Lemma~\ref{lemma:path-implies-execution}) and \emph{complete}
(Lemma~\ref{lemma:execution-implies-path}). \emph{Note that we first
  establish these results assuming that the \wPN $\Nn$ given as
  parameter is an \wOPN, then prove that the results extend to the
  general case of \wPN}.

\paragraph{Soundness} To establish \emph{soundness} of our algorithm,
we show that, for every path $n_0,\ldots, n_k$ in the tree returned by
$\treealgo(\Nn,m_0)$, and for every target marking
$m\in\gamma(\lambda(n_k))$, we can find an execution of $\Nn$ reaching
a marking $m'\in\gamma(n_k)$ that covers $m$. This implies that, if
$\lambda(n_k)(p)=\w$ for some $p$, then, we can find a family of
executions that reach a marking in $\gamma(n_k)$ with an arbitrary
number of tokens in $p$. For instance, consider the path $n_1,n_2,n_3$
in $\Tt_1'$ (Fig.~\ref{fig:example-tree}), and let $m=\tuple{0,
  2,4}$. Then, a corresponding execution is
$\tuple{1,0,0}\xrightarrow{t_1}\tuple{0,4,0}\xrightarrow{t_2}\tuple{0,3,2}\xrightarrow{t_2}\tuple{0,2,4}$. Remark
that the execution is not necessarily the sequence of transitions
labeling the path in the tree: in this case, we need to iterate $t_2$
to transfer tokens from $p_2$ to $p_3$, which is summarised in one
edge $(n_2,n_3)$ in $\Tt_1$, by the acceleration.

\begin{lemma}\label{lemma:path-implies-execution}
  Let $\Nn$ be an \wOPN, let $m_0$ be an \w-marking and let $\Tt$ be
  the tree returned by $\treealgo(\Nn,m_0)$. Let $\pi=n_0,\ldots, n_k$
  be a stuttering path in $\Tt$, and let $m$ be a marking in
  $\gamma(\lambda(n_k))$. Then, there exists an execution
  $\rho_\pi=m_0\xrightarrow{t_1}m_1\cdots\xrightarrow{t_\ell}m_\ell$
  of $\Nn$ s.t. $m_\ell\in\gamma(\lambda(n_k))$, $m_\ell\mgeq m$ and
  $m_0\in\gamma(\lambda(n_0))$.
  Moreover, when for all $0\leq i\leq j\leq k$: $\nbw{n_i}=\nbw{n_j}$,
  we have: $t_1\cdots t_\ell=\mu(\pi)$.
\end{lemma}

\paragraph{Completeness} Proving completeness amounts to showing that
every execution (starting from $m_0$) of an \wPN $\Nn$ is witnessed by
a stuttering path in $\treealgo(\Nn,m_0)$. It relies on the following
property:
\begin{lemma}\label{lemma:every-node-is-fully-developed-or-has-no-succ}
  Let $\Nn$ be an \wOPN, let $m_0$ be an \w-marking, and let $\Tt$ be
  the tree returned by $\treealgo(\Nn,m_0)$. Then, for all nodes $n$
  of $\SCT(\Nn,m_0)$:
  \begin{itemize}
  \item either $n$ has no successor in the tree and has an ancestor
    $\nbar$ s.t. $\lambda(\nbar)=\lambda(n)$.
  \item or the set of successors of $n$ corresponds to all the
    $\rightarrow_\omega$ possible successors of $\lambda(n)$, i.e.:
    $\{\mu(n,n')\mid (n,n')\in E\} = \{t\mid
    \lambda(n)\xrightarrow{t}_\omega\}$. Moreover, for each $n'$
    s.t. $(n,n')\in E$ and $\mu(n,n')=t$:
    $\lambda(n')\mgeq\lambda(n)+\effect(t)$.
  \end{itemize}
\end{lemma}

We can now state the completeness property:
\begin{lemma}\label{lemma:execution-implies-path}
  Let $\Nn$ be an \wOPN with set of transitions $T$, let $m_0$ be an
  initial marking and let
  $m_0\xrightarrow{t_1}m_1\xrightarrow{t_2}\cdots\xrightarrow{t_n}m_n$
  be an execution of $\Nn$. Then, there are a \emph{stuttering} path
  $\pi=n_0,n_1,\ldots,n_k$ in $\SCT(\Nn,m_0)$ and a monotonic
  increasing mapping $h:\{1,\ldots,n\}\mapsto\{0,\ldots,k\}$ s.t.:
  $\mu(\pi) = t_1t_2\cdots t_n$ and  $m_i\mleq\lambda(n_{h(i)})$ for
  all $0\leq i\leq n$.
\end{lemma}

\paragraph{From \wOPN to \wPN} We have shown completeness and
soundness of the \treealgo algorithm for \wOPN. Let us show that each
\wPN $\Nn$ can be turned into an \wOPN $\remIw(\Nn)$ that $(i)$
terminates iff $\Nn$ terminates and $(ii)$ that has the same
coverability sets as $\Nn$. The \wOPN $\remIw(\Nn)$ is obtained from
$\Nn$ by replacing each transition $t\in T$ by a transition $t'\in T'$
s.t. $O(t')=O(t)$ and $I(t')=\{I(t)(p)\otimes p\mid
I(t)(p)\neq\omega\}$. Intuitively, $t'$ is obtained from $t$ by
deleting all $\w$ input arcs. Since $t'$ always consumes \emph{less
  tokens} than $t$ does, the following is easy to establish:

\begin{lemma}\label{lem-rem-I-w}
  Let $\Nn$ be an \wPN. For all executions $m_0,t_1',m_1,\ldots,
  t_n',m_n$ of $\remIw(\Nn)$: $m_0,t_1,m_1,\ldots, t_n,m_n$ is an
  execution of $\Nn$. For all finite (resp. infinite) executions $m_0,
  t_1, m_1,\ldots, t_n, m_n$ ($m_0, t_1, m_1,\ldots, t_j,m_j,\ldots$)
  of $\Nn$, there exists an execution $m_0, t_1', m_1',\ldots,
  t_n',m_n'$ ($m_0, t_1, m_1',\ldots, t_j,m_j',\ldots$) of
  $\remIw(\Nn)$, s.t. $m_i\mleq m_i'$ for all $i$.
\end{lemma}

Intuitively, this means that, when solving coverability, (place)
boundedness or termination on an \wPN $\Nn$, we can analyse
$\remIw(\Nn)$ instead, because $\Nn$ terminates iff $\remIw(\Nn)$
terminates, and removing the \w-labeled input arcs from $\Nn$ does not
allow to reach higher markings. Finally, we observe that, for all \wPN
$\Nn$, and all initial marking $m_0$: the trees returned by
$\treealgo(\Nn,m_0)$ and $\treealgo\left(\remIw(\Nn,m_0)\right)$
respectively are isomorphic\footnote{That is, if $\treealgo(\Nn,m_0)$
  returns $\tuple{N,E,\lambda, \mu, n_0}$ and
  $\treealgo\left(\remIw(\Nn,m_0)\right)$
  returns$\tuple{N',E',\lambda', \mu', n_0'}$, then, there is a
  bijection $h:N\mapsto N'$ s.t. $(i)$ $h(n_0)=n_0'$, $(ii)$ for all
  $n\in N$: $\lambda(n)=\lambda(h(n))$, $(iii)$ for all $n_1$, $n_2$
  in $N$: $(n_1, n_2)\in E$ iff $(h(n_1), h(n_2))\in E'$, $(iv)$ for
  all $(n_1, n_2)\in E$: $\mu(n_1, n_2)=\mu'(h(n_1), h(n_2))$.}.  This
is because we have defined $c-\w$ to be equal to $c$: applying this
rule when computing the effect of a transition $t$
(line~\ref{lst:effect-trans}), is equivalent to computing the effect
of the corresponding $t'$ in $\remIw(\Nn)$, i.e. letting $I(t')(p)=0$
for all $p$ s.t. $I(t)(p)=\w$. Thus, we can lift
Lemma~\ref{lemma:path-implies-execution} and
Lemma~\ref{lemma:execution-implies-path} to \wPN. This establish
correctness of the algorithm for the general \wPN case.


\paragraph{Applications of the \treename tree}
These results allow us to conclude that the \treename can be used to
compute a coverability set and to decide termination of any \wPN.
\begin{theorem}\label{theo:correctness-km}
  Let $\Nn$ be an \wPN with initial marking $m_0$, and let $\Tt$ be
  the tree returned by
  $\tuple{N,E,\lambda,\mu,n_0}=\treealgo(\Nn,m_0)$. Then: $(i)$
  $\lambda(N)$ is a coverability set of $\Nn$ and
  $(ii)$~$\Nn$~terminates iff $\Tt$ contains an $\omega$-maximal
  self-covering stuttering path.
\end{theorem}
\begin{proof}
  Point $(i)$ follows from Lemma~\ref{lemma:path-implies-execution}
  (lifted to \wPN). Let us now prove both directions of point $(ii)$.

  First, we show that if $\treealgo(\Nn,m_0)$ contains an
  $\omega$-maximal self-covering stuttering path, then $\Nn$ admits a
  self-covering execution from $m_0$. Let $n_0,\ldots, n_k,\break
  n_{k+1},\ldots,n_\ell$ be an $\omega$-maximal self-covering
  stuttering path, and assume\break
  $\effect(n_{k+1},\ldots,n_\ell)\geq \mathbf{0}$. Let us apply
  Lemma~\ref{lemma:path-implies-execution} (lifted to \wPN), by
  letting $m=\mathbf{0}$ and $\pi=\pi_2$, and let $m_1$ and $m_2$ be
  markings s.t. $m_1\xrightarrow{\mu(\pi_2)}m_2$. The existence of
  $m_1$ and $m_2$ is guaranteed by
  Lemma~\ref{lemma:path-implies-execution} (lifted to \wPN), because
  all the nodes along $\pi_2$ have the same number of $\omega$'s as we
  are considering an \emph{$\omega$-maximal} self-covering stuttering
  path. Since $\effect(\pi_2)$ is positive, so is
  $\effect(\mu(\pi_2))$. Thus, there exists\footnote{Remark that,
    although $\effect(\mu(\pi_2))\mgeq\mathbf{0}$, we have no
    guarantee that $m_2\mgeq m_1$, as we could have
    $\effect(\mu(\pi_2))=\omega$ for some $p$, and maybe the amount of
    tokens that has been produced in $p$ by $\mu(\pi_2)$ to yield
    $m_2$ does not allow to have $m_2(p)\geq m_1(p)$. However, in this
    case, it is always possible to reach a marking with enough tokens
    in $p$ to cover $m_1(p)$, since $\effect(\mu(\pi_2))=\omega$. }
  $m_2'$ s.t. $m_1\xrightarrow{\mu(\pi_2)}m_2'$ and $m_2'\mgeq
  m_1$. By invoking Lemma~\ref{lemma:path-implies-execution} (lifted
  to \wPN) again, letting $\pi=\pi_1$ and $m=m_1$, we conclude to the
  existence of a sequence of transitions $\sigma$, a marking $m_0$ and
  a marking $m_1'\mgeq m_1$ s.t. $m_0\xrightarrow{\sigma}m_1'$. Since
  $m_1'\mgeq m_1$, $\mu(\pi_2)$ is again firable from $m_1'$. Let
  $\mbar_2=m_2+m_1'-m_1$. Clearly,
  $m_1'\xrightarrow{\mu(\pi_2)}\mbar_2$, with $\mbar_2\mgeq
  m_1'$. Hence,
  $m_0\xrightarrow{\sigma}m_1'\xrightarrow{\mu(\pi_2)}\mbar_2$ is a
  self-covering execution of $\Nn$.

  Second, let us show that, if $\Nn$ admits a self-covering execution
  from $m_0$, then $\treealgo(\Nn,m_0)$ contains an $\omega$-maximal
  self-covering stuttering path. Let
  $\rho=m_0\xrightarrow{t_1}m_1\cdots \xrightarrow{t_n}m_n$ be a
  self-covering execution and assume $0\leq k<n$ is a position
  s.t. $m_k\mleq m_n$. Let $\sigma_1$ denote $t_1,\ldots t_k$ and
  $\sigma_2$ denote $t_{k+1},\ldots t_n$. Let us consider the
  execution $\rho'$, defined as follows
  \begin{align*}
    \rho' &= m_0\xrightarrow{\sigma_1}m_k
    \underbrace{\xrightarrow{t_{k+1}}m_{k+1}\cdots\xrightarrow{t_n}m_n}_{\sigma_2}
    \underbrace{\xrightarrow{t_{k+1}}m_{n+1}\cdots\xrightarrow{t_n}m_{2n-k}}_{\sigma_2}
    \cdots\\
    &\phantom{=}
    \cdots\underbrace{\xrightarrow{t_{k+1}}m_{(|P|+1)n-|P|k+1}\cdots\xrightarrow{t_n}m_{(|P|+2)n-(|P|+1)k}}_{\sigma_2}
  \end{align*}
  where for all $n+1\leq j\leq (|P|+2)n-(|P|+1)k$: $m_j - m_{j-1} =
  m_{f(j)}-m_{f(j-1)}$ with $f$ the function defined as $f(x)=
  \big((x-k)\mod(n-k)\big)+k$ for all $x$. Intuitively, $\rho'$
  amounts to firing $\sigma_1(\sigma_2)^{|P|+1}$ (where $P$ is the set
  of places of $\Nn$) from $m_0$, by using, each time we fire
  $\sigma_2$, the same effect as the one that was used to obtain
  $\rho$ (remember that the effect of $\sigma_2$ is non-deterministic
  when $\omega$'s are produced). It is easy to check that $\rho'$ is
  indeed an execution of $\Nn$, because $\rho$ is a self-covering
  execution.

  Let $n_0,n_1,\ldots n_\ell$ and $h$ be the stuttering path in
  $\treealgo(\Nn,m_0)$ and the mapping corresponding to $\rho'$ (and
  whose existence is established by
  Lemma~\ref{lemma:execution-implies-path}). Since, $m_k\mleq m_n$,
  $\effect(t_{k+1}\cdots t_n)\geq \mathbf{0}$ and by
  Lemma~\ref{lemma:execution-implies-path} (lifted to \wPN), all the
  following stuttering paths are self-covering:
  \begin{align*}
    &n_0,\ldots,n_{h(k)},\ldots, n_{h(n)}\\
    &n_0,\ldots,n_{h(k)},\ldots, n_{h(n)},\ldots,n_{h(2n-k)}\\
    &n_0,\ldots,n_{h(k)},\ldots, n_{h(n)},\ldots,n_{h(2n-k)},\ldots,n_{h(3n-2k)}\\
    &\vdots\\
    &n_0,\ldots,n_{h(k)},\ldots, n_{h(n)},\ldots,n_{h(2n-k)},\ldots,n_{h(3n-2k)},\ldots,n_{h((|P|+2)n-(|P|+1)k)}\\\
  \end{align*}
  Let us show that one of them is $\omega$-maximal, i.e. that there is
  $1\leq j\leq |P|+1$
  s.t. $\nbw{n_{h(jn-(j-1)k)}}=\nbw{n_{h((j+1)n-jk)}}$. Assume it is
  not the case. Since the number of $\omega$'s can only increase along
  a stuttering path, this means that 
  \begin{align*}
    &0\leq \nbw{n_{h(n)}}<\nbw{n_{h(2n-k)}}<\nbw{n_{h(3n-2k)}}<\nbw{n_{h((|P|+2)n-(|P|+1)k)}}
  \end{align*}
  However, this implies that $\nbw{n_{h((|P|+2)n-(|P|+1)k)}}>|P|$,
  which is not possible as $P$ is the set of places of $\Nn$. Hence,
  we conclude that there exists an $\omega$-maximal self-covering
  stuttering path in  $\treealgo(\Nn,m_0)$.\qed
\end{proof}


\section{From \wPN to plain PN\label{sec:solv-cover-place}\label{sec:reach}}
Let us show that we can, from any \wPN $\Nn$, build a plain PN $\Nn'$
whose set of reachable markings allows to recover the reachability set
of $\Nn$. This construction allows to solve reachability, coverability
and (place) boundednes. The idea of the construction is depicted in
Fig.~\ref{fig:wPN2PN}. More precisely, we turn the \wPN
$\Nn=\tuple{P,T,m_0}$ into a plain PN $\Nn'=\tuple{P', T', m_0'}$
using the following procedure. Assume that $T=T_{plain}\uplus
T_\omega$, where $T_\omega$ is the set of $\omega$-transitions of
$\Nn$. Then:
\begin{enumerate}
\item We add to the net one place (called the \emph{global lock})
  $\mathsf{lock_g}$, and for each $\omega$-transition~$t$, one place
  $\mathsf{lock}_t$. That is,
  $P'=P\cup\{\mathsf{lock_g}\}\cup\{\mathsf{lock}_t\mid t\in
  T_\omega\}$.
\item Each transition $t$ in $\Nn$ is replaced by a set of transitions
  $T_t$ in $\Nn'$. In the case where $t$ is a plain transition, $T_t$
  contains a single transition that has the same effect as $t$, except
  that it also tests for the presence of a token in
  $\mathsf{lock_g}$. In the case where $t$ is an $\omega$-transition,
  $T_t$ is a set of plain transitions that simulate the effect of $t$,
  as in Fig.~\ref{fig:wPN2PN}. Formally, $T'=\cup_{t\in T} T_t$, where
  the $T_t$ sets are defined as follows:
  \begin{itemize}
  \item If $t$ is a plain transition, then $T_t=\{t'\}$, where,
    $I(t')=I(t)\cup\{\mathsf{lock_g}\}$ and
    $O(t')=O(t)\cup\{\mathsf{lock_g}\}$.
  \item If $t$ is an $\omega$-transition, then:
    \begin{align*}
      T_t &= \left\{t', t_{end}\right\}\cup \{t^p_{-\omega}\mid
      I(t)(p)=\omega\}\cup \{t^p_{+\omega}\mid O(t)(p)=\omega\}
    \end{align*}
    where $I(t')=I(t)+\{\mathsf{lock_g}\}$;
    $O(t')=I(t_{end})=\{\mathsf{lock}_t\}$;
    $O(t_{end})=\{\mathsf{lock_g}\}+ O(t)$. Furthermore, for all $p$
    s.t. $I(t)(p)=\omega$: $I(t^p_{-\omega})=\{p, \mathsf{lock}_t\}$
    and $O(t^p_{-\omega})=\{\mathsf{lock}_t\}$. Finally, for all $p$
    s.t. $O(t)(p)=\omega$: $I(t^p_{+\omega})=\{ \mathsf{lock}_t\}$ and
    $O(t^p_{-\omega})=\{p,\mathsf{lock}_t\}$.
  \end{itemize}
\item We let $f$ be the function that associates each marking $m$ of
  $\Nn$ to the marking $f(m)$ of $\Nn'$ s.t. $m'(\mathsf{lock_g})=1$;
  for all $p\in P$: $m'(p)=m(p)$; and for all $p\not\in
  P\cup\{\mathsf{lock_g}\}$: $m'(p)=0$. Then, the initial marking of
  $\Nn'$ is $f(m_0)$.
\end{enumerate}

\begin{figure}
  \centering
  \begin{tikzpicture}
    \begin{scope}

      \node[transition] (t) {$t$} ;

      \node[place] (q) [above left= of t]{$q$} ; 
      \node[place] (q1) [left= of t] {$q_1$} ; 
      \node[place] (q2) [below left= of t] {$q_2$} ;

      \node[place] (p) [above right= of t] {$p$} ; 
      \node[place] (p1) [right= of t] {$p_1$} ; 
      \node[place] (p2) [below right = of t] {$p_2$} ;
    
      \path 
      (t) edge [pre] node {} (q) 
          edge [pre] node [above] {$\omega$} (q1) 
          edge [pre] node [above=.1cm] {$\omega$} (q2) 
          edge [post] node {} (p) 
          edge [post] node [above] {$\omega$} (p1) 
          edge [post] node [above=.1cm] {$\omega$} (p2) 
    ;
    \end{scope}
    \begin{scope}[xshift=6.5cm]
      \node[place] (lockt) {$\mathsf{lock}_t$} ;
  
      \node[transition] (tp) [above left=of lockt]{$t'$} ;
      \node[transition] (tend) [above right= of lockt] {$t_{end}$} ;

      \node[place] (q) [left = of tp] {$q$} ; 
      \node[place] (lockg) [above= of lockt] {$\mathsf{lock_g}$} ; 
      \node[place] (p) [right = of tend] {$p$} ;

      \path (tp) edge [pre] (q) edge [pre] (lockg) edge [post]
      (lockt);

      \path (tend) edge [pre] (lockt) edge [post] (lockg) edge [post]
      (p);

      \node[transition] (to1) [ right= of lockt] {$t_{+\omega}^{p_1}$} ; 
      \node[transition] (to2) [below right= of lockt] {$t_{+\omega}^{p_2}$} ; 
      \node[transition] (toc1) [ left= of lockt] {$t_{-\omega}^{q_1}$} ; 
      \node[transition] (toc2) [below left= of lockt] {$t_{-\omega}^{q_2}$} ;
  
      \foreach \x in {1,2} { 
        \node[place] (p\x) [right=of to\x] {$p_\x$} ; 
        \node[place] (q\x) [left=of toc\x] {$q_\x$} ;
 
        \path (to\x) edge [pre, bend left=15] (lockt) 
                     edge [post, bend right=15] (lockt) 
                     edge [post] (p\x) ;
        
        \path (toc\x) edge [pre, bend left=15] (lockt) 
                      edge [post, bend right=15] (lockt) 
                      edge [pre] (q\x) ; 
      }
 
    \end{scope}

  \end{tikzpicture}
  \caption{Transforming an \wPN into a plain PN.}
  \label{fig:wPN2PN}
\end{figure}
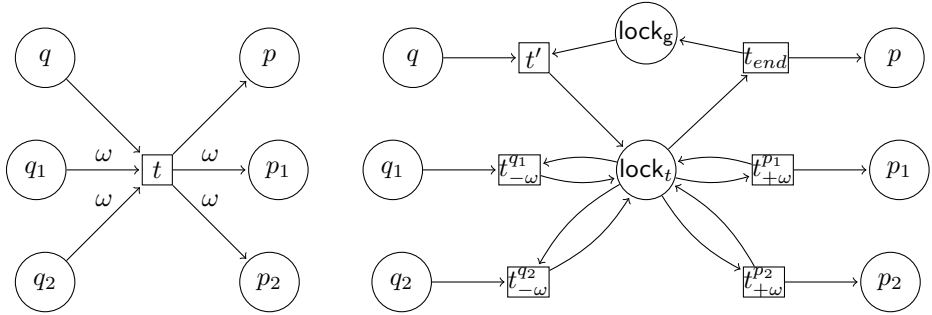

It is easy to check that:
\begin{lemma}\label{lem:encoding-in-plain-pn}
  Let $\Nn$ be an \wPN and let $\Nn'$ be its corresponding PN. Then
  $m\in\reach(\Nn)$ iff $f(m)\in\reach(\Nn')$.
\end{lemma}
The above construction can be carried out in polynomial time. Thus,
\wPN generalise Petri nets, the known complexities for reachability
\cite{lipton63,May84}, (place) boundedness and coverability
\cite{Rackoff78} carry on to \wPN:
\begin{corollary}
  Reachability for \wPN is decidable and
  \textsc{ExpSpace}-hard. Coverability, boundedness and place
  boundedness for \wPN are \textsc{ExpSpace}-c.
\end{corollary}
This justifies the result given in Table~\ref{tab:complexity} for
reachability, coverability and (place) boundedness, for \wPN.

However, the above construction fails for deciding termination. For
instance, assume that the leftmost part of Fig.~\ref{fig:wPN2PN} is an
\wPN $\Nn=\tuple{P,T,m_0}$ with $m_0(q)=1$. Clearly, all executions of
$\Nn$ are finite, while $t'(t^{p_1}_{+\omega})^\omega$ is an infinite
transition sequence that is firable in $\Nn'$. Termination, however is
decidable, by the \treeshortname technique of
Section~\ref{sec:solv-term}, and \textsc{ExpSpace}-hard, as \wPN
generalise Petri nets. In the next section, we show that the Rackoff
technique \cite{Rackoff78} can be generalised to \wPN, and prove that
termination is \textsc{ExpSpace}-c for \wPN.


\section{Extending the Rackoff technique for \wPN\label{sec:extend-rack-techn}}
In this section, we extend the Rackoff technique to \wPN to prove the
existence of short self-covering sequences. For applications of
interest, such as the termination problem, it is sufficient to
consider \wOPN, as proved in Lemma~\ref{lem-rem-I-w}. Hence, we only
consider \wOPN in this section.

As observed in \cite{Rackoff78}, beyond some large values, it is not
necessary to track the exact value of markings to solve some problems.
We use threshold functions $\threshold: \{0, \ldots, |P|\} \to \NN$ to
specify such large values. Let $\nbwb{m} = |\{p \in P \mid m(p) \in
\NN\}|$.
\begin{definition}
  \label{def:thresholdMarkings}
  Let $\threshold: \{0, \ldots, |P|\} \to \NN$ be a threshold
  function. Given an $\omega$-marking $m$, the markings
  $\ceilthr{m}{\threshold}$ and $\floorthr{m}{\threshold}$ are defined
  as follows:
  \begin{align*}
    (\ceilthr{m}{\threshold}) ( p) & =
    \begin{cases}
      m( p) & \text{if } m(p) < \threshold( \nbwb{m}),\\
      \omega & \text{otherwise.}
    \end{cases}
    \\
    ( \floorthr{m}{\threshold}) (p) & =
    \begin{cases}
      m(p) & \text{if } m(p) \in \NN,\\
      \threshold( \nbwb{m} + 1) & \text{otherwise.}
    \end{cases}
  \end{align*}
\end{definition}
In $\ceilthr{m}{\threshold}$, values that are too high are abstracted
by $\omega$. In $\floorthr{m}{\threshold}$, $\omega$ is replaced by
the corresponding natural number. This kind of abstraction is
formalized in the following threshold semantics.
\begin{definition}
  \label{def:thresholdPN}
  Given an \wPN $\Nn$, a transition $t$, an $\omega$-marking $m$ that
  enables $t$ and a threshold function $\threshold$, we define the
  transition relation $\tstep{t}{\threshold}$ as $m
  \tstep{t}{\threshold} \ceilthr{m + \effect(t)}{\threshold}$.
\end{definition}
The transition relation $\tstep{t}{\threshold}$ is extended to
sequences of transitions in the usual way. 
Note that if $m \tstep{t}{\threshold} m'$, then
$\cbasis{m} \subseteq \cbasis{m'}$. In words, a place marked $\omega$
will stay that way along any transition in threshold semantics.

Let $\maxred = \max\{|\effect(t)(p)| \mid t\in T, p \in P,
\effect(t)(p) < \omega\}$. The following proposition says that
$\omega$ can be replaced by natural numbers that are large enough so
that sequences are not disabled. The proof is by a routine induction
on the length of sequences, using the fact that in an \wOPN, any
transition can reduce at most $\maxred$ tokens from any place.
\begin{proposition}
  \label{prop:thSemToNatSem}
  For some $\omega$-markings $m_{1}$ and $m_{2}$, suppose $m_{1}
  \tstep{\seq}{\threshold} m_{2}$ and $\cbasis{m_{2}} =
  \cbasis{m_{1}}$. If $m_{1}'$ is a marking such that $m_{1}'
  \mleqp{\cbasis{m_{1}}} m_{1}$ and $m_{1}' ( p) \ge \maxred |\seq|$
  for all $p \in \cbasis{m_{1}}$, then $m_{1}' \step{\seq} m_{2}'$
  such that $m_{2}' \mleqp{\cbasis{m_{2}}} m_{2}$ and $m_{2}' ( p) \ge
  m_{1}'(p) - \maxred |\seq|$.
\end{proposition}

\begin{definition}
  \label{def:thPumpSeq}
  Given an $\omega$-marking $m_{1}$ and a threshold function
  $\threshold$, an \emph{$\omega$-maximal threshold
  pumping sequence} ($\threshold$-PS) enabled at $m_{1}$ is a
  sequence $\seq$ of transitions such that $m_{1}
  \tstep{\seq}{\threshold} m_{2}$, $\eff( \seq) \ge \0$ and
  $\cbasis{m_{2}} = \cbasis{m_{1}}$.
\end{definition}
In the above definition, note that we require $\effect( \seq) ( p) \ge
0$ for \emph{any} place $p$, irrespective of whether $m_{1} (p) = \omega$
or not.
\begin{definition}
  \label{def:loops}
  Suppose $\seq$ is an $\omega$-maximal $\threshold$-PS enabled at
  $m_{1}$ and $\seq = \seq_{1} \seq_{2} \seq_{3}$ such that $m_{1}
  \tstep{\seq_{1}}{\threshold} m_{3} \tstep{\seq_{2}}{\threshold}
  m_{3} \tstep{\seq_{3}}{\threshold} m_{2}$. We call $\seq_{2}$ a
  \emph{simple loop} if all intermediate $\omega$-markings obtained
  while firing $\seq_{2}$ from $m_{3}$ (except the last one, which is
  $m_{3}$ again) are distinct from one another.
\end{definition}
In the above definition, since $m_{3} \tstep{\seq_{2}}{\threshold}
m_{3}$ and $m_{1} \tstep{\seq_{1} \seq_{3}}{ \threshold} m_{2}$, one
might be tempted to think that $\seq_{1}\seq_{3}$ is also an
$\omega$-maximal $\threshold$-PS enabled at $m_{1}$. This is however
not true in general, since there might be some $p \in \cbasis{m_{1}}$
such that $\effect( \seq_{1}\seq_{3}) ( p) < 0$ (which is compensated
by $\seq_{2}$ with $\effect( \seq_{2}) (p) >0$). The presence of the
simple loop $\seq_{2}$ is required due to its compensating effect. The
idea of the proof of the following lemma is that if there are a large
number of loops, it enough to retain a few to get a shorter
$\omega$-maximal $\threshold$-PS.

\begin{lemma}
  \label{lem:shortThPumpSeq}
  There is a constant $d$ such that for any \wPN{} $\Nn$, any
  threshold function $\threshold$ and any $\omega$-maximal
  $\threshold$-PS $\seq$ enabled at some $\omega$-marking $m_{1}$,
  there is an $\omega$-maximal $\threshold$-PS $\seq'$ enabled at
  $m_{1}$, whose length is at most $(\threshold(\nbwb{m_{1}}) 2
  \maxred)^{d |P|^{3}}$.
\end{lemma}
\begin{proof}[Sketch]
  This proof is similar to that of \cite[Lemma 4.5]{Rackoff78}, with
  some modifications to handle $\omega$-transitions. It is organized
  into the following steps.
  \begin{list}{Step}{}
    \item[Step 1:] We first associate a vector with a sequence of
      transitions to measure the effect of the sequence. This is the
      step that differs most from that of \cite[Lemma 4.5]{Rackoff78}.
      The idea in this step is similar to the one used in \cite[Lemma
      7]{BJK2010}.
    \item[Step 2:] Next we remove some simple loops from $\seq$ to
      obtain $\seq''$ such that for every intermediate
      $\omega$-marking $m$ in the run $m_{1} \tstep{\seq}{\threshold}
      m_{2}$, $m$ also occurs in the run $m_{1}
      \tstep{\seq''}{\threshold} m_{2}$.
    \item[Step 3:] The sequence $\seq''$ obtained above need not be a
      $\threshold$-PS. With the help of the vectors defined in step 1,
      we formulate a set of linear Diophantine equations that encode
      the fact that the effects of $\seq''$ and the simple loops that
      were removed in step 2 combine to give the effect of a
      $\threshold$-PS.
    \item[Step 4:] Then we use the result about existence of small solutions to
      linear Diophantine equations to construct a sequence $\seq'$
      that meets the length constraint of the lemma.
    \item[Step 5:] Finally, we prove that $\seq'$ is a $\threshold$-PS
      enabled at $m_{1}$.
  \end{list}
  
  \emph{Step 1}: Let $P_{\omega} \subseteq \cbasis{m_{1}}$ be the set
  of places $p$ such that some transition $t$ in $\seq$ has
  $\effect(t) ( p) = \omega$. If we ensure that for each place $p \in
  P_{\omega}$, some transition $t$ with $\effect (t) (p) = \omega$ is
  fired, we can ignore the effect of other transitions on $p$. This is
  formalized in the following definition of the effect of any sequence
  of transitions $\seq_{1} = t_{1} \cdots t_{r}$. We define the
  function $\effabs{P_{\omega}}{\seq_{1}}: \cbasis{m_{1}} \to \ZZ$ as
  follows.
  \begin{align*}
    \effabs{P_{\omega}}{\seq_{1}}(p) =
    \begin{cases}
      1 & p\in P_{\omega}, \exists i\in \{1, \ldots, r\}: \effect(t_{i})(p) = \omega\\
      0 & p\in P_{\omega}, \forall i\in \{1, \ldots, r\}: \effect(t_{i})(p) \ne \omega\\
      \sum_{1 \le i \le r} \effect ( t_{i})(p) & \text{otherwise}
    \end{cases}
  \end{align*}
  Applying the above definition to simple loops, it is possible to
  remove some of them to get shorter pumping sequences. Details about
  how to do it are in the remaining steps of the proof, which are
  moved to the appendix.\qed
\end{proof}

\begin{definition}
  \label{def:ThFns}
  Let $c = 2 d$. The functions $ \threshold_{1}, \threshold_{2},
  \lenfn: \NN \to \NN$ are as follows:
  \begin{align*}
    \threshold_{1}(0) & =  1 &
    \lenfn(0) & = (2 \maxred)^{c |P|^{3}} & \threshold_{2} ( 0 ) & =
    \maxred\\
    \threshold_{1} (i+1) & = 2 \maxred \lenfn(i) & \lenfn(i + 1) & =
    (\threshold_{1} (i + 1)  2 \maxred)^{c |P|^{3}} &
    \threshold_{2} ( i + 1) & = \maxred \lenfn(i)
  \end{align*}
\end{definition}
All the above functions are non-decreasing. Due to the selection of
the constant $c$ above, we have $(2 x \maxred)^{c |P|^{3}} \ge x^{|P|}
+ ( 2 x \maxred)^{d |P|^{3}}$ for all $x \in \NN$.

\newlength{\myl}
\setlength{\myl}{1cm}

\tikzstyle{connect}=[->,rounded corners=0.2\ml]

\tikzstyle{place}=[shape=circle,draw=black,inner sep=0,
minimum size=0.6\myl,transform shape]

\usetikzlibrary{decorations}
\usepgflibrary{decorations.pathreplacing}

The goal is to prove that if there is a self-covering execution, there
is one whose length is at most $\lenfn(|P|)$. That proof uses the
result of Lemma~\ref{lem:shortThPumpSeq} and the definition of
$\lenfn$ above reflects it. For the intuition behind the definition of
$\threshold_{1}$ and $\threshold_{2}$, suppose that the proof of the
length upper bound of $\lenfn(|P|)$ is by induction on $|P|$ and we
have proved the result for $|P| = i$. For the case of $i+1$, we want
to decide the value beyond which it is safe to abstract by replacing
numbers by $\omega$.
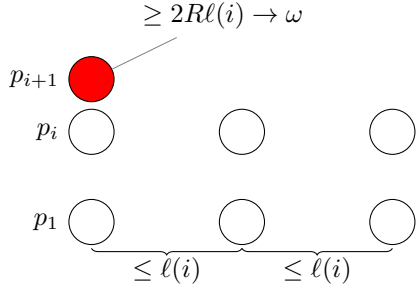
\begin{figure}[!t]
 \centerline{%
    \begin{tikzpicture}[>=stealth]
  \node[place, label=180:$p_{1}$] (p11) at (0\myl,0\myl) {};
  \node[place, label=180:$p_{i}$] (p1i) at ([yshift=1.2\myl]p11) {};
  \node[place, fill=red, label=180:$p_{i+1}$, pin=45:$\ge 2 \maxred \lenfn(i) \to \omega$] (p1i1) at ([yshift=0.7\myl]p1i) {};

  \node[place] (p21) at ([xshift=2\myl]p11) {};
  \node[place] (p2i) at (p21 |- p1i) {};

  \node[place] (p31) at ([xshift=2\myl]p21) {};
  \node[place] (p3i) at (p31 |- p1i) {};

  \draw[decorate, decoration=brace, auto=left] ([yshift=-0.35\myl]p21.center) to node {$\le \lenfn(i)$} ([yshift=-0.35\myl]p11.center);
  \draw[decorate, decoration=brace, auto=left] ([yshift=-0.35\myl]p31.center) to node {$\le \lenfn(i)$} ([yshift=-0.35\myl]p21.center);
\end{tikzpicture}

}
  \caption{Intuition for the threshold functions}
  \label{fig:ThFnIntuition}
\end{figure}
As shown in Fig.~\ref{fig:ThFnIntuition}, suppose the initial prefix
of a self-covering execution for $i$ places is of length at most
$\lenfn(i)$. Also suppose the pumping portion of the self-covering
execution is of length at most $\lenfn(i)$. The total length is at
most $2 \lenfn(i)$. Since each transition can reduce at most
$\maxred$ tokens from any place, it is enough to have $2 \maxred
\lenfn(i)$ tokens in $p_{i+1}$ to safely replace numbers by $\omega$.

The following lemma shows that if some $\omega$-marking can be reached
in threshold semantics, a corresponding marking can be reached in the
natural semantics where $\omega$ is replaced by a value large enough
to solve the termination problem.
\begin{lemma}
  \label{lem:thCovToNatCov}
  For some $\omega$-markings $m_{3}$ and $m_{4}$, suppose $m_{3}
  \tstep{\seq}{\threshold_{1}} m_{4}$. Then there is a sequence
  $\seq'$ such that $\floorthr{m_{3}}{\threshold_{1}} \step{\seq'}
  m_{4}'$, $m_{4}' \mgeqp{\cbasis{m_{4}}}
  \floorthr{m_{4}}{\threshold_{2}}$ and $|\seq'| \le
  \threshold_{1}(\nbwb{m_{3}})^{|P|}$.
\end{lemma}

\begin{lemma}
  \label{lem:shortSCS}
  If an \wPN $\Nn$ admits a self-covering execution, then it admits
  one whose sequence of transitions is of length at most $\lenfn( |P|
  )$.
\end{lemma}
\begin{proof}
  Suppose $\seq = \seq_{1}\seq_{2}$ is the sequence of transitions in
  the given self-covering execution such that $m_{0}
  \step{\seq_{1}} m_{1} \step{\seq_{2}} m_{2}$ and $m_{2} \mgeq
  m_{1}$. A routine
  induction on the length of any sequence of transitions $\seq$ shows
  that if $m_{3} \step{\seq} m_{4}$, we have $m_{3}
  \tstep{\seq}{\threshold_{1}} m_{4}'$ with $m_{4}' - m_{3} \mgeq
  m_{4} - m_{3}$. Hence, we have $m_{0}
  \tstep{\seq_{1}}{\threshold_{1}} m_{1}'
  \tstep{\seq_{2}}{\threshold_{1}} m_{2}'$ with $m_{2}' \mgeq m_{1}'$.
  By monotonicity, we infer that for any $i \in \NNP$, $m_{i}'
  \tstep{\seq_{2}}{\threshold_{1}} m_{i+1}'$ with $m_{i+1}' \mgeq
  m_{i}'$. Let $j \in \NNP$ be the first number such that
  $\cbasis{ m_{j}'} = \cbasis{m_{j+1}'}$. We have $m_{0}
  \tstep{\seq_{1} \seq_{2}^{j-1}}{\threshold_{1}} m_{j}'
  \tstep{\seq_{2}}{\threshold_{1}} m_{j+1}'$ and $\seq_{2}$ is an
  $\omega$-maximal $\threshold_{1}$-PS enabled at $m_{j}'$.

  By Lemma~\ref{lem:shortThPumpSeq}, there is a $\threshold_{1}$-PS
  $\seq_{2}'$ enabled at $m_{j}'$ whose length is at most\break
  $(\threshold_{1}( \nbwb{m_{j}'}) 2 \maxred)^{d |P|^{3}}$. By
  Lemma~\ref{lem:thCovToNatCov}, there is a sequence $\seq_{1}'$ such
  that $m_{0} \step{\seq_{1}'} m_{j}''$, $m_{j}''
  \mgeqp{\cbasis{m_{j}'}} \floorthr{m_{j}'}{\threshold_{2}}$ and
  $|\seq_{1}'| \le (\threshold_{1}( |P|))^{|P|}$. By
  Definition~\ref{def:ThFns} and
  Definition~\ref{def:thresholdMarkings}, we infer that $m_{j}'' ( p)
  = \maxred \lenfn( \nbwb{m_{j}'}) = \maxred ( \threshold_{1}(
  \nbwb{m_{j}'}) 2 \maxred)^{c |P|^{3}} \ge \maxred |\seq_{2}'|$ for
  all $p \in \cbasis{m_{j}'}$. Hence, we infer from
  Proposition~\ref{prop:thSemToNatSem} that $m_{0} \step{\seq_{1}'}
  m_{j}'' \step{\seq_{2}'} m_{j+1}''$. Since $\seq_{2}'$ is a
  $\threshold_{1}$-PS, $\effect(\seq_{2}') \mgeq \0$, and so
  $m_{j+1}'' \mgeq m_{j}''$. Therefore, firing $\seq_{1}' \seq_{2}'$
  at $m_{0}$ results in a self-covering execution. The length of
  $\seq_{1}' \seq_{2}'$ is at most $(\threshold_{1}( |P|))^{|P|} +
  (\threshold_{1}( \nbwb{m_{j}'}) 2 \maxred)^{d |P|^{3}} \le \lenfn(
  |P|)$.\qed
\end{proof}

\begin{lemma}
  \label{lem:GrowthFnUpBnd}
  Let $k = 3 c$. Then $\lenfn( i) \le (2
  \maxred)^{k^{i + 1} |P|^{3 ( i + 1 ) }}$ for all $i
  \in \NN$.
\end{lemma}

\begin{theorem}
  \label{thm:terminationExpsp}
  The termination problem for \wPN is \textsc{ExpSpace}-c.
\end{theorem}
The idea of the proof of the above theorem is to construct a
non-deterministic Turing machine that guesses and verifies a
self-covering sequence. By Lemma~\ref{lem:shortSCS}, the length of
such a sequence can be limited and hence made to work in
\textsc{ExpSpace}. Full proof can be found in the appendix.



\section{Extensions with transfer or reset arcs\label{sec:extensions}}
In this section, we consider two extensions of \wPN, namely: \wPN with
\emph{transfer arcs} (\wPNT) and \wPN with \emph{reset arcs}
(\wPNR). These extensions have been considered in the case of plain
Petri nets: Petri nets with transfer arcs (\PNT) and Petri nets with
reset arcs (\PNR) have been extensively studied in the literature
\cite{DFS98,ACJT96,DJS99,S10}. Intuitively, a \emph{transfer arc}
allows, when the corresponding transition is fired to \emph{transfer
  all the tokens} from a designated place $p$ to a given place $q$,
while a \emph{reset arc} \emph{consumes all tokens} from a designated
place $p$.

Formally, an \emph{extended} \wPN is a tuple $\tuple{P,T}$, where $P$
is a finite set of places and $T$ is finite set of transitions. Each
transition is a pair $t=(I,O)$ where $I:P\mapsto \mathbb{N}\cup\{\w,
\mathsf{T}, \mathsf{R}\}$; $O:P\mapsto \mathbb{N}\cup\{\w,
\mathsf{T}\}$; $|\{p\mid I(p)\in\{\mathsf{T},\mathsf{R}\}\}|\leq 1$;
$|\{p\mid O(p)\in\{\mathsf{T}\}\}|\leq 1$; there is $p$
s.t. $I(p)=\mathsf{T}$ \emph{iff} there is $q$ s.t. $O(q)=\mathsf{T}$;
and \emph{if} there is $p$ s.t. $I(p)=\mathsf{R}$, \emph{then},
$O(p)\in\mathbb{N}\{\w\}$ for all $p$. A transition $(I,O)$
s.t. $I(p)=\mathsf{T}$ (resp. $I(p)=\mathsf{R}$) for some $p$ is
called a \emph{transfer} (\emph{reset}). An \wPN with \emph{transfer
  arcs} (resp. \emph{with reset arcs}), \wPNT (\wPNR) for short, is an
extended \wPN that contains no reset (transfer). An \wPNT
s.t. $I(t)(p)\neq \w$ for all transitions $t$ and places $p$ is an
\wOPNT . The class \wIPNT is defined symmetrically. An \wPNT which is
both an \wOPNT \emph{and} an \wIPNT is a (plain) \PNT. The classes
\wOPNR, \wIPNR and \PNR are defined accordingly.

Let $t=(I,O)$ be a transfer or a reset. $t$ is \emph{enabled} in a
marking $m$ iff for all $p$: $I(p)\not\in\{\w,\mathsf{T},\mathsf{R}\}$
implies $m(p)\geq I(p)$. In this case firing $t$ yields a marking
$m'=m-m_I+m_O$ (denoted $m\xrightarrow{t}m'$) where for all $p$:
$m_I(p)=m(p)$ if $I(p)\in\{\mathsf{T},\mathsf{R}\}$; $0\leq m_I(p)\leq
m(p)$ if $I(p)=\w$; $m_I(p)=I(p)$ if $I(p)\not\in \{\mathsf{T},
\mathsf{R}, \w\}$; $m_O(p)=m(p')$ if $O(p)=I(p')=\mathsf{T}$ ;
$m_O(p)\geq 0$ if $O(p)=\w$; and $m_O(p)=O(p)$ if $O(p)\not\in
\{\mathsf{T}, \w\}$. The semantics of transitions that are neither
transfers nor resets is as defined for \wPN.

Let us now investigate the status of the problems listed in
Section~\ref{sec:omega-petri-nets}, in the case of \wPNT and
\wPNR. First, since \wPNT (\wPNR) extend \PNT (\PNR), the lower bounds
for the latters carry on: reachability and place-boundedness are
undecidable \cite{Duf98} for \wPNT and \wPNR; boundedness is
undecidable for \wPNR \cite{DJS99}; and coverability is Ackerman-hard
for \wPNT and \wPNR \cite{S10}. On the other hand, the construction
given in Section~\ref{sec:solv-cover-place} can be adapted to turn an
\wPNT (resp. \wPNR) $\Nn$ into a \PNT (\PNR) $\Nn'$ satisfying
Lemma~\ref{lem:encoding-in-plain-pn} (i.e., projecting
$\reach(\Nn',m_0)$ on the set of places of $\Nn$ yields $\reach(\Nn,
m_0)$). Hence, boundedness for \wPNT \cite{DJS99}, and coverability
for both \wPNT and \wPNR are decidable \cite{ACJT96}.

As far as \emph{termination} is concerned, it is decidable
\cite{DFS98} and Ackerman-hard \cite{S10} for \PNR and
\PNT. Unfortunately, the construction presented in
Section~\ref{sec:solv-cover-place} does not preserve termination, so
we cannot reduce termination of \wPNT (resp. \wPNR) to termination of
\PNT (\PNR). Actually, termination becomes undecidable when
considering \wOPNR or \wOPNT:
\begin{theorem}\label{theo:termination-undecidable-for-wPNT-wPNR}
  Termination is undecidable for \wOPNT and \wOPNR with one
  $\w$-output-arc
\end{theorem}
\begin{proof}
  We first prove undecidability for \wOPNT. The proof is by reduction
  from the \emph{parameterised termination problem} for
  \emph{Broadcast protocols} (BP) \cite{EFM99}. It is well-known that
  \PNT generalise broadcast protocols, hence the following
  \emph{parameterised termination problem for \PNT} is
  \emph{undecidable}: `given a \PNT $\tuple{P,T}$ and an $\w$-marking
  $\mbar_0$, does $\tuple{P,T,m_0}$ terminate \emph{for all}
  $m_0\in\dc{\mbar_0}$ ?'  From a \PNT $\Nn=\tuple{P,T}$ and an
  \w-marking $\mbar_0$, we build the \wOPNT (with only one
  \w-output-arc) $\Nn'=\tuple{P',T',m_0'}$ where
  $P'=P\uplus\{p_{init}\}$, $T'=T\uplus\{(I,O)\}$, $I=\{p_{init}\}$,
  $O=\{\w\otimes p\mid \mbar_0(p)=\w\}$, and $m_0'=\{\mbar_0\otimes
  p\mid \mbar_0(p)\neq\w\}$. Clearly, $\Nn'$ terminates iff
  $\tuple{P,T,m_0}$ terminates for all $m_0\in\dc{\mbar_0}$. Hence,
  termination for \wOPNT is \emph{undecidable} too. Finally, we can
  transform an \wOPNR $\Nn=\tuple{P,T,m_0}$ into an \wOPNT
  $\Nn'=\tuple{P\uplus\{p_{trash}\}, T', m_0}$, where $t'\in T'$ iff
  either $(i)$ $t'\in T$ and $t'$ is not a reset, or $(ii)$ there is a
  reset $t\in T$ and a place $p\in P$ s.t. $I(t)(p)=\mathsf{R}$,
  $I(t')(p)=\mathsf{T}$, $O(t')(p_{trash})=\mathsf{T}$, for all
  $p'\neq p$: $I(t')(p')=I(t)(p')$ and for all $p''\neq p_{trash}$:
  $O(t')(p'')=O(t)(p'')$. Intuitively, the construction replaces each
  reset (resetting place $p$) in $\Nn$ by a transfer from $p$ to
  $p_{trash}$ in $\Nn'$, where $p_{trash}$ is a fresh place from which
  no transition consume. Since $\Nn'$ terminates iff $\Nn$ terminates,
  termination is undecidable for \wPNR too. \qed
\end{proof}
However, the construction of Section~\ref{sec:solv-cover-place} can be
applied to \wIPNT and \wIPNR to yield a corresponding \PNT
(resp. \PNR) that preserves termination. Hence, termination is
decidable and Ackerman-hard for those models. This justifies the
results on \wPNT and \wPNR given in Table~\ref{tab:complexity}.



\clearpage
\appendix
\section{Proof of Lemma~\ref{lemma:terminates-iff-no-s-c-e}}
\noindent\textit{ An \wPN terminates iff it admits no self-covering
  execution.}
\begin{proof}
  Assume $\Nn=\tuple{P,T,m_0}$ admits an infinite execution
  $m_0\xrightarrow{t_1}m_1\rightarrow{t_2}\cdots\xrightarrow{t_j}
  m_j\xrightarrow{t_{j+1}}\cdots$. Since $\mleq$ is a well-quasi
  ordering on the markings, there are two positions $\alpha$ and
  $\beta$ in the execution s.t. $\alpha\leq\beta$ and $m_\alpha\mleq
  m_\beta$. Hence,
  $m_0\xrightarrow{t_1}m_1\xrightarrow{t_2}\cdots\xrightarrow{t_\beta}m_\beta$
  is a self-covering execution.

  For the reverse implication, assume $\Nn=\tuple{P,T,m_0}$ admits a
  self-covering execution
  $m_0\xrightarrow{t_1}m_1\rightarrow{t_2}\cdots\xrightarrow{t_n}m_n$
  and assume $0\leq k<n$ is a position s.t. $m_k\mleq m_n$. Then, by
  monotonicity, it is possible to fire infinitely often the
  $t_{k+1}\cdots t_n$ sequence from $m_k$. More precisely, one can
  check that the following is infinite execution of $\Nn$:
  $$
  \begin{array}{c}
    m_0
    \xrightarrow{t_1}m_1
    \cdots
    \xrightarrow{t_k}m_k
    \xrightarrow{t_{k+1}}m^0_{k+1}
    \cdots
    \xrightarrow{t_n}m^0_n
    \xrightarrow{t_{k+1}}m^1_{k+1}
    \cdots
    \xrightarrow{t_n}m^1_n\\
    \xrightarrow{t_{k+1}}m^2_{k+1}
    \cdots
    \xrightarrow{t_n}m^2_{n}
    \cdots
    \xrightarrow{t_{k+1}}m^j_{k+1}
    \cdots
    \xrightarrow{t_n}m^j_{n}
    \cdots
  \end{array}
  $$
  where for all $1\leq i\leq n-k$: $m^0_{k+i}=m_{k+i}$, for all $j\geq
  1$, $m^j_{k+1}=m^{j-1}_{n}+(m_{k+1}-m_k)$ and for all $2\leq i\leq
  n-k$: $m^j_i=m^j_{i-1}+(m_{k+i}-m_{k+i-1})$.\qed
\end{proof}
\section{Proof of Proposition~\ref{prop:termination} (Termination)}
 \noindent\textit{For all \wPN $\Nn$ and for all initial marking $m_0$,
  $\SCT(\Nn,m_0)$ terminates.}
\begin{proof}
  The proof is by contradiction. Assume $\SCT(\Nn,m_0)$ does not
  terminate. First observe that the recursion depth is always bounded:
  since a recursive call is performed only when a new $\omega$ has
  been created, the recursion depth is, at any time, at most equal to
  $|P|+1$, where $P$ is the set of places of $\Nn$

  Thus, if $\SCT(\Nn,m_0)$ does not terminate, it is necessarily
  because the main \textbf{while} loop does not terminate (the other
  loop of the algorithm is the \textbf{forall} starting in
  line~\ref{lst:forall_t}, which always execute at most $|T|$
  iterations, where $T$ is the set of transitions of $\Nn$). In this
  loop, one node is removed from \texttt{U} at each iteration. Since
  the algorithm builds a tree, a node that has been removed from
  \texttt{U} will never be inserted again in \texttt{U}. Hence, the
  tree $\Tt$ built by $\SCT(\Nn,m_0)$ is infinite.

  By K\"onig's lemma, and since $\Tt$ is finitely branching, it
  contains an infinite path $\pi$. Since the recursion depth is
  bounded, $\pi$ can be split into a finite prefix $\pi_1$ and an
  infinite suffix $\pi_2$ s.t. all the nodes in $\pi_2$ have been
  built during the same recursive call.

  Let us assume $\pi_2=n_0,n_1,\ldots,n_m,\ldots$  Since $\mleq$ is a
  well-quasi-ordering on $\omega$-markings, there are $k$ and $\ell$
  s.t. $0\leq k<\ell$ and $\lambda(n_k)\mleq\lambda(n_\ell)$. Clearly,
  $\lambda(n_k)=\lambda(n_\ell)$ is not possible because of the test
  of line~\ref{lst:if2} that prevents the development of $n_\ell$ in
  this case. Thus, $\lambda(n_k)\ml\lambda(n_\ell)$. This means that,
  for all $p\in P$: $\lambda(n_k)(p)\leq\lambda(n_\ell)(p)$, and that
  there exists $p$ s.t. $\lambda(n_k)(p)<\lambda(n_\ell)(p)$. Let
  $p^<$ be such a place. By definition of the {\tt Post} function, and
  of the acceleration (line~\ref{lst:def-m-w}), the only possibility
  is that $\lambda(n_\ell)(p^<)=\w\neq\lambda(n_k)(p^<)$. However, in
  this case, when $\lambda(n_\ell)$ is returned by \texttt{Post}, a
  new recursive call is triggered, which contradicts the hypothesis
  that $n_\ell$ and $n_k$ have been built during the same recursive
  call. Contradiction.  \qed
\end{proof}

\section{Proof of Lemma~\ref{lemma:path-implies-execution}
  (soundness)}
Recall that, in the present section, we prove the soundness of
\treealgo, \emph{when applied to \wOPN only}. Hence, throughout the
section $I(t)(p)\neq\omega$ for all places $p$ and transitions $t$. To
prove Lemma~\ref{lemma:path-implies-execution}, we need ancillary
results and definitions. First, we state the \emph{place monotonicity}
property of \wPN.  Let $m_1$ and $m_2$ be two markings, and let
$P'\subseteq P$ be a set of places s.t. $m_2\mgeqp{P'}m_1$. Let
$\sigma$ be a sequence of transitions and let $m_3$ be a
marking\footnote{Remark that, due to the $\omega$'s, the effect of
  $\sigma$ is now non-deterministic, and there can be several such
  $m_3$.}  s.t. $m_1\xrightarrow{\sigma}m_3$. Then, there exists a
marking $m_4$ s.t. $m_2\xrightarrow{\sigma}m_4$ and
$m_4\mgeqp{P'}m_3$.

Then, we observe, that, when no $\omega$'s are introduced in the
labels of the nodes, the sequence of labels along a branch coincides
with the effect of the transitions labelling this branch. Formally:
\begin{lemma}\label{lemma:no-omega-implies-exact-effect}
  Let $\Nn$ be an \wOPN, let $m_0$ be an \w-marking and let $\Tt$ be
  the tree returned by $\SCT(\Nn,m_0)$. Let $n_1$, $n_2$ be two nodes
  of $\Tt$ s.t. $(n_1,n_2)\in E^+$. 
  Then, for all $p$ s.t. $\lambda(n_1)(p)\neq\omega$ and
  $\lambda(n_2)(p)\neq\omega$, we have:
  $\lambda(n_2)(p)=\lambda(n_1)(p)+\effect(\sigma)(p)$.
\end{lemma}

The next technical definitions allows to characterise when a sequence
of transition is firable from a given marking. Let $\sigma=t_1\cdots
t_n$ be a sequence of transitions of an \wOPN, s.t. for all $1\leq
i\leq n-1$, for all $p\in P$: $O(t_i)(p)\neq \omega$. Let $m$ be a
marking and let $p$ be a place.  Then, we let $\af$ be the predicate
s.t. $\af(\sigma, m, p)$ is true iff:
\begin{align*}
  \forall 1\leq i\leq n &: m(p)+\effect(t_1\cdots t_{i-1})(p) \geq I(t_i)(p)
\end{align*}
Remark that $\sigma$ is firable from $m$ iff for all $p\in P$:
$\af(\sigma, m, p)$.  We extend the definition of $\af$ to sequences
of transitions containing one $\omega$-output-transition. Let
$\sigma=t_1\cdots t_n$ be a sequence of transitions, let $p$ be a
place, and let $1\leq j\leq n$ be the least position
s.t. $O(t_j)(p)=\omega$. Then $\af(\sigma, m, p)$ holds iff
$\af(t_1\cdots t_j,m,p)$ holds. Again, $\sigma$ is firable from $m$
iff for all $p\in P$: $\af(\sigma, m, p)$. Indeed, $\af(t_1\cdots
t_j,m,p)$ ensures that, when firing $\sigma$ from $m$, $p$ will never
be negative along $t_1\cdots t_j$. Moreover, $t_j$ can create an
arbitrary large number of tokens in $p$, since $O(t_j)(p)=\omega$,
which allows to ensure that $p$ will never be negative along
$t_{j+1}\cdots t_n$. Given this definition of $\af$ it is easy to
observe that:
\begin{enumerate}
\item $m(p)\geq I(\sigma)(p)$ implies that $\af(\sigma, m, p)$,
\item if $\af(\sigma,m,p)$ holds and $\effect(\sigma)(p)\geq 0$, then
  $\af(\sigma^K,m,p)$ holds too for all $K\geq 1$.
\end{enumerate}

\begin{lemma}\label{lemma:soundness}
  Let $\Nn$ be an \wOPN, let $m_0$ be an \w-marking, and let $\Tt$ be
  the tree returned by $\treealgo(\Nn, m_0)$, let $e=(n_1,n_2)$ be an
  edge of $\Tt$ and let $m$ be a marking in
  $\gamma(\lambda(n_2))$. Then, there are
  $m_1\in\gamma(\lambda(n_1))$, $m_2\in\gamma(\lambda(n_2))$ and a
  sequence of transitions $\sigma_\pi$ of $\Nn$
  s.t. $m_1\xrightarrow{\sigma_\pi} m_2$ and $m_2\mgeq m$. Moreover,
  when $\nbw{\lambda(n_1)}=\nbw{\lambda(n_2)}$, $\sigma_\pi=\mu(e)$ is
  a sequence of transitions meeting these properties.
\end{lemma}
\begin{proof}
  Edges are created by \treealgo in line~\ref{lst:connect} only. Thus,
  by the test of the \texttt{forall} loop (line~\ref{lst:forall_t}), and
  since we are considering an \wOPN:
%
  \begin{eqnarray}
    \label{eq:3}
    \lambda(n_1)&\geq&I(\mu(e))
  \end{eqnarray}
  Moreover, when creating an edge $(n,n')$ (line~\ref{lst:connect}),
  $n'$ is either a fresh node s.t. $\lambda(n')$ is the \w-marking
  returned by $\Post(\Nn,\lambda(n), t)$, or $n'$ is the root of the
  subtree $\Tt'$ returned by the recursive call $\treealgo(\Nn, m')$,
  with $\mu(n,n')=t$ in both cases. However, in the latter case, the
  root of $\Tt'$ is $m'$, i.e., the marking returned by
  $\Post(\Nn,\lambda(n), t)$ too. Since this holds for all edges, we
  conclude that $\lambda(n_2)$ is the \w-marking $m'$ returned by
  $\Post(\Nn, \lambda(n_1), \mu(e))$. Considering the definition of
  the $\Post$ function, we see that $m'$ is either
  $\lambda(n_1)-I(t)+O(t)$ (when the condition of the \texttt{if} in
  line~\ref{lst:if-post} is not satisfied), or the result $m_\w$ of an
  acceleration (when the condition of the \texttt{if} in
  line~\ref{lst:if-post} is satisfied). We consider these two cases
  separately.

  \noindent {\bf {\sc Case A}: the condition of the $\mathtt{if}$ in
    line~\ref{lst:if-post} has not been satisfied (i.e., no
    acceleration has occurred).} Then, $\lambda(n_2)$ is the marking
  $m'$ computed in line~\ref{lst:effect-trans}:
  \begin{align}
    \lambda(n_2)&=\lambda(n_1)-I(\mu(e))+O(\mu(e))\label{eq:4}
  \end{align}
  We let $m_1$ be the marking s.t. for all places $p\in P$:
  \begin{align*}
    m_1(p) &=
    \begin{cases}
      \lambda(n_1)(p) &\textrm{if }\lambda(n_1)(p)\neq\omega\\
      I(\mu(e))(p)+m(p)&\textrm{otherwise}
    \end{cases}
  \end{align*}
  And we let $m_2$ be the marking s.t., for all places $p\in P$: 
  \begin{align*}
    m_2(p) &=
    \begin{cases}
      m_1(p)+O(\mu(e))(p)-I(\mu(e))(p) &\textrm{if }O(\mu(e))(p)\neq \omega\\
      m_1(p)-I(\mu(e))(p)+m(p)&\textrm{otherwise}
    \end{cases}
  \end{align*}
  Finally, we let:
  \begin{align*}
    \sigma_\pi &= \mu(e)
  \end{align*}
  Let us show that $m_1$, $m_2$ and $\sigma_\pi=\mu(e)$ satisfy the
  lemma. First, we observe that $m_1\in\gamma(\lambda(n_1))$, by
  definition. Then, we further observe that there are only four
  possibilities regarding the possible values of $\lambda(n_1)(p)$,
  $\lambda(n_2)(p)$ and $O(\mu(e))(p)$, as shown in the following
  table. Indeed, $n_2$ is a successor of $n_1$ in the tree, so
  $\pom{n_2}\supseteq \pom{n_1}$. Moreover,
  $\lambda(n_2)(p)=\omega\neq \lambda(n_1)(p)$ holds for some $p$ iff
  $O(\mu(e))(p)=\omega$, as we have assumed that the condition of the
  $\mathtt{if}$ in line~\ref{lst:if-post} has not been satisfied:
  $$
  \begin{array}{c||c|c|c}
    \textrm{Case}  &\lambda(n_1)(p)&\lambda(n_2)(p)&O(\mu(e))(p)\\
    \hline\hline
    1& =\omega& =\omega& =\omega\\\hline  
    2& =\omega& =\omega& \neq\omega\\\hline
    3& \neq\omega&=\omega&=\omega\\\hline
    4& \neq\omega&\neq\omega&\neq\omega\\
  \end{array}
  $$
  For these four different cases, we obtain the following values for
  $m_1(p)$ and $m_2(p)$, by definition:
  \begin{align}
    m_1(p) &=
    \begin{cases}
      I(\mu(e))(p)+m(p)&\textrm{cases 1 and 2}\\
      \lambda(n_1)(p)&\textrm{cases 3 and 4}
    \end{cases}\label{eq:5}
  \end{align}    
  
  \begin{align}
    m_2(p) &=
    \begin{cases}
      2\times m(p)&\textrm{case 1}\\
      m(p)+O(\mu(e))(p)&\textrm{case 2}\\
      \lambda(n_1)(p)-I(\mu(e))(p)+m(p)&\textrm{case 3}\\
      \lambda(n_1)(p)+O(\mu(e))(p)-I(\mu(e))(p)&\textrm{case 4}
    \end{cases}\label{eq:6}
  \end{align}    

  To prove that $m_2\in\gamma(\lambda(n_2))$, we must show that
  $m_2(p)=\lambda(n_2)(p)$ for all $p$ s.t. $\lambda(n_2)(p)\neq
  \omega$, which corresponds only to case 4, where we have: 
  \begin{align*}
    m_2(p)&= \lambda(n_1)(p)+O(\mu(e))(p)-I(\mu(e))(p)&\textrm{By~(\ref{eq:6})}\\
    &= \lambda(n_2)(p)&\textrm{By~(\ref{eq:4})}
  \end{align*}
 
  Then, it remains to show that $m_1\xrightarrow{\mu(e)} m_2$. First,
  we show that, $\mu(e)$ is firable from $m_1$, i.e. that for all
  $p\in P$: $m_1(p)\geq I(\mu(e))(p)$. In case 1 and 2, we have
  $m_1(p)=I(\mu(e))(p)+m(p)\geq I(\mu(e))(p)$. In cases 3 and 4, we
  have $m_1(p)=\lambda(n_1)(p)$, with $\lambda(n_1)(p)\geq
  I(\mu(e))(p)$ by (\ref{eq:3}). Thus, $\mu(e)$ is firable from
  $m_1$. Then, we must show that $m_2$ can be obtained as a successor
  of $m_1$ by $\mu(e)$. In cases 1 and 3, the effect of $\mu(e)$ is to
  remove $I(\mu(e))(p)$ tokens from $p$ and to produce an arbitrary
  number $K$ of tokens in $p$. Hence, in case~1, by firing $\mu(e)$
  from $m_1$, we obtain $I(\mu(e))(p)+m(p)-I(\mu(e))(p)+K=m(p)+K$
  tokens in $p$. In case~3, by firing $\mu(e)$ from $m_1$, we obtain
  $\lambda(n_1)(p)-I(\mu(e))(p)+K$ tokens in $p$. In both cases, by
  letting $K=m(p)$, we obtain $m_2(p)$. In cases 2 and 4, the effect
  of $\mu(e)$ on place $p$ is equal to
  $O(\mu(e))(p)-I(\mu(e))(p)$. Hence, in case 2, by firing $\mu(e)$
  from $m_1$, we obtain
  $I(\mu(e))(p)+m(p)-I(\mu(e))(p)+O(\mu(e))(p)=m(p)+O(\mu(e))(p)$
  tokens in $p$. In case 4, by firing $\mu(e)$ from $m_1$, we obtain
  $\lambda(n_1)(p)-I(\mu(e))(p)+O(\mu(e))(p)$ tokens in $p$. In both
  cases, these values correspond exactly to $m_2(p)$.

  We conclude this case by observing that
  $\nbw{\lambda(n_1)}=\nbw{\lambda(n_2)}$ implies that no acceleration
  has been performed, which is the present case. We have thus shown
  that when $\nbw{\lambda(n_1)}=\nbw{\lambda(n_2)}$,
  $\sigma_\pi=\mu(e)$ is a sequence of transitions that satisfies the
  lemma.  \medskip

  \noindent{\bf {\sc Case B}: the condition of the $\mathtt{if}$ in
    line~\ref{lst:if-post} has been satisfied (an acceleration has
    occurred)}. Remark that, in this case, $n_1$ is the node called $n$
  in the condition of the \textbf{if}, and $\mu(e)$ is the transition
  called $t$ in the same condition.  Let $\sigmabar$ be the sequence
  of transitions labelling the path from $\nbar$ to $n_1$. Let
  $P^{Acc}$ denote the set of places:
  \begin{align}
    P^{Acc} &= \{p\mid \effect(\sigmabar(p))>0\wedge
    \lambda(n_2)(p)\neq\omega\wedge O(\mu(e))(p)\neq\omega\}\label{eq:def-P-acc}
  \end{align}
  Then, let $K$ be the value defined as:
  \begin{align}
    K &= \max_{p\in P^{Acc}}\{m(p)\}\label{eq:def-K}
  \end{align}
  This value allows us to define the sequence of transitions $\sigma_\pi$:
  \begin{align}
    \sigma_\pi &= \mu(e)\big(\sigmabar\cdot\mu(e)\big)^K\label{eq:def-sigma-pi}
  \end{align}
  
  From those definitions of $\nbar$, $n_1$, $n_2$, $\sigmabar$ and
  $\mu(e)$, we conclude that only the following cases are possible,
  for all places $p$:
  $$
  \begin{array}{c||cccccc}
    \mathrm{case}& \lambda(\nbar)(p)& \lambda(n_1)(p)& \lambda(n_2)(p)& \effect(\sigmabar)(p)& \effect(\mu(e))(p)& \mathrm{Remark}\\
    \hline\hline
    1            & \omega           & \omega         & \omega         & \in\mathbb{Z}\cup\{\omega\}& \in\mathbb{Z}\cup\{\omega\}\\
    2            & \neq \omega      & \neq \omega    & \neq\omega     & \neq \omega                &\neq \omega\\
    3            & \neq \omega      & \neq \omega    & \omega         & \neq \omega                & \omega\\
    4            & \neq \omega      & \neq \omega    & \omega         & \neq \omega                & \neq\omega & \effect(\sigmabar\cdot\mu(e))(p)>0
  \end{array}
  $$
  Those cases are the only possible because $\nbar$ is an ancestor of
  $n_1$, which is itself an ancestor of $n_2$. Moreover, by
  construction, $\nbw{\nbar}=\nbw{n_1}$, since those two nodes have
  been computed during the same recursive call. Thus, the occurrence
  of a fresh $\omega$ can only appear between $n_1$ and $n_2$, either
  because $\effect(\mu(e))(p)=\omega$ (case 3), or because we have
  performed an acceleration (case 4). Remark that the latter only
  occurs when $\effect(\sigmabar\cdot\mu(e))(p)>0$.
   
  Let us next define the marking $m_1$, as:
  \begin{eqnarray}
    m_1(p)&=&
    \left\{
      \begin{array}{ll}
        \lambda(n_1)(p)&\textrm{if }\lambda(n_1)(p)\neq \omega\\
        I(\sigma_\pi)(p)+m(p)&\textrm{otherwise}
      \end{array}
    \right.\label{eq:def-of-m1}
  \end{eqnarray}
  where $I(\sigma_\pi)(p)$ denotes $\sum_{i=1}^nI(t_i)(p)$ for
  $\sigma_\pi=t_1,\ldots, t_n$.  Observe that, by definition: $m_1\in
  \gamma(\lambda(n_1))$. %
  Then, let us prove that $\sigma_\pi$ is firable from $m_1$. First
  observe that, if $p$ is a place s.t. $\lambda(n_1)(p)=\omega$, then
  $\af(\sigma_\pi, m_1, p)$ holds, because, in this case, $m_1(p)\geq
  I(\sigma_\pi)(p)$, by (\ref{eq:def-of-m1}). Then, assume $p$ is a
  place s.t. $\lambda(n_1)(p)\neq\omega$. In this case, by definition,
  $m_1(p)=\lambda(n_1)$. First observe that, by construction, and
  since we consider \wOPN (see line~\ref{lst:forall_t} of the
  algorithm):
  \begin{align}
    \forall p: \lambda(n_1)(p) &\geq I(\mu(e))(p)\label{eq:8}
  \end{align}
  Let us now consider all the possible cases, which are cases 2, 3 and
  4 from the table above (case 1 cannot occur since we have assumed
  that $\lambda(n_1)(p)\neq\omega$):
  \begin{itemize}
  \item \emph{In case 2}, since the condition of the \textbf{if}
    (line~\ref{lst:if-post}) is satisfied, we know that
    $\effect(\sigmabar\cdot \mu(e))(p)\geq 0$. Since
    $\lambda(\nbar)(p)\neq\omega$, and $\lambda(n_1)(p)\neq\omega$, we
    can apply Lemma~\ref{lemma:no-omega-implies-exact-effect}, and
    conclude that:
    \begin{align*}
      \lambda(n_2)(p)&=\lambda(\nbar)(p)+\effect(\sigmabar\cdot\mu(e))(p)\\
      &= \lambda(\nbar)(p)+\effect(\sigmabar)(p) + \effect(\mu(e))(p)\\
      &= \lambda(n_1)(p)+\effect(\mu(e))(p)
    \end{align*}
    Thus:
    \begin{align}
      \lambda(n_1)(p)+\effect(\mu(e))(p) &\geq \lambda(\nbar)(p)\label{eq:7}
    \end{align}
    since $\effect(\sigmabar\cdot \mu(e))(p)\geq 0$. By applying
    \textsc{Case A} (above) iteratively along the branch from $\nbar$
    to $n_1$, we deduce that $\af(\sigmabar,\lambda(\nbar),p)$
    holds. Hence,
    $\af(\sigmabar,\lambda(n_1)(p)+\effect(\mu(e))(p),p)$ holds too,
    by (\ref{eq:7}). Finally, by~(\ref{eq:8}), we conclude that
    $\af(\mu(e)\cdot\sigmabar,\lambda(n_1)(p),p)$ holds. However,
    $\effect(\mu(e)\cdot\sigmabar)(p)=\effect(\sigmabar\cdot
    \mu(e))(p)\geq 0$. Thus, since $\mu(e)\cdot\sigmabar$ has a
    positive effect on $p$, we conclude that
    $\af\left((\mu(e)\cdot\sigmabar)^K,\lambda(n_1)(p),p\right)$ holds
    too, for all $K\geq 1$. Finally, since
    $\effect\left((\mu(e)\cdot\sigmabar)^K\right)(p)\geq 0$, we
    conclude that
    $$
    \lambda(n_1)(p)+\effect\left((\mu(e)\cdot\sigmabar)^K\right)\geq
    \lambda(n_1)(p)
    $$
    Thus, by~(\ref{eq:8}), we have
    $$
    \lambda(n_1)(p)+\effect\left((\mu(e)\cdot\sigmabar)^K\right)\geq
    I(\mu(e))
    $$ 
    and we can thus fire $\mu(e)$ once again after firing
    $(\mu(e)\cdot\sigmabar)^K$. Hence,
    $$
    \af\left((\mu(e)\cdot\sigmabar)^K\cdot\mu(e),\lambda(n_1),p\right)
    $$
    holds, with $\sigma_\pi=(\mu(e)\cdot\sigmabar)^K\cdot\mu(e)$.
  \item \emph{In case 3}: by~(\ref{eq:8}), since
    $O(\mu(e))(p)=\omega$, and since $\mu(e)$ is the first transition
    of $\sigma_\pi$, we immediately conclude that $\af(\sigma_\pi,
    \lambda(n_1),p)$.
  \item \emph{In case 4}, we can adapt the reasoning of case 2 as
    follows. First remember, that, in case 4,
    $\effect(\sigmabar\cdot\mu(e))(p)>0$.  Since
    $\lambda(\nbar)(p)\neq\omega$, and $\lambda(n_1)(p)\neq\omega$, we
    can apply Lemma~\ref{lemma:no-omega-implies-exact-effect}, and
    conclude that
    $\lambda(n_1)(p)=\lambda(\nbar)(p)+\effect(\sigmabar)(p)$. Thus:
    \begin{align*}
      \lambda(n_1)(p)+\effect(\mu(e))(p) &=\lambda(\nbar)(p)+\effect(\sigmabar)(p)+\effect(\mu(e))(p)\\
      &=\lambda(\nbar)(p)+\effect(\sigmabar\cdot\mu(e))(p)
    \end{align*}
    with $\effect(\sigmabar\cdot\mu(e))(p)>0$. Hence:
    \begin{align*}
      \lambda(n_1)(p)+\effect(\mu(e))(p)&>\lambda(\nbar)(p)
    \end{align*}
    This implies~(\ref{eq:7}), and we can thus reuse the arguments of
    case 2 to conclude that
    $\af\left(\sigma_\pi,\lambda(n_1),p\right)$ holds in the present
    case too.
  \end{itemize}
  Thus, for all $p$ s.t. $\lambda(n_1)(p)\neq\omega$:
  $\af(\sigma_\pi,\lambda(n_1),p)$ holds. However,
  $\lambda(n_1)(p)\neq\omega$ implies that $m_1(p)=\lambda(n_1)(p)$,
  hence, $\af(\sigma_\pi, m_1,p)$ holds in those cases. Thus, we conclude
  that $\af(\sigma_\pi,m_1,p)$ holds for all places $p$, and thus, that
  $\sigma_\pi$ is firable from $m_1$.

  To conclude the proof let us build a marking $m_2$ that respects the
  conditions given in the statement of the lemma. Let $\mbar$ be a
  marking s.t. $m_1\xrightarrow{\sigma_\pi}\mbar$. We know that such a
  marking exists since $\sigma_\pi$ is firable from $m_1$. We first
  observe that, by
  Lemma~\ref{lemma:when-no-omega-the-effect-is-the-same-as-PN}:
  \begin{align}
    \forall p\textrm{ s.t. }\effect(\sigma_\pi)(p)\neq\omega &:
    \mbar(p)=m_1(p)+\effect(\sigma_\pi)(p)\label{eq:10}
  \end{align}
  From $\mbar$,
  we define $m_2$ as follows:
  \begin{align}\label{eq:9}
    m_2(p) &=
    \begin{cases}
      \mbar(p)&\textrm{if }\effect(\sigma_\pi)(p)\neq\w\\
      \max\left\{\mbar(p),m(p)\right\}&\textrm{otherwise}
    \end{cases}
  \end{align}
  Clearly, $m_2\mgeqp{P'}\mbar$, for
  $P'=\{p\mid\effect(\sigma_\pi)(p)=\w$. Hence, by
  Lemma~\ref{lemma:reachable-markings-is-upward-closed},
  $m_1\xrightarrow{\sigma_\pi}m_2$ holds. Let us conclude the proof by
  showing that $m_2\in\gamma(\lambda(n_2))$, and that $m_2\geq m$, as
  requested. Since $m$ has been assumed to be in
  $\gamma(\lambda(n_2))$ too, it is sufficient to show that for all
  place $p$: $(i)$ $\lambda(n_2)(p)=\omega$ implies $m_2(p)\geq m$,
  and $(ii)$ $\lambda(n_2)(p)\neq\omega$ implies
  $m_2(p)=\lambda(n_2)(p)$.

  Thus, we consider each place $p$ separately, by reviewing the four
  cases given in the table above:
  \begin{itemize}
  \item \emph{In case 1}, $m_1(p)=I(\sigma_\pi)(p)+m(p)$ and
    $\lambda(n_2)(p)=\omega$. Let us show that $m_2(p)\geq m(p)$. We
    consider two further cases:
    \begin{enumerate}
    \item either $\effect(\sigma_\pi)(p)\neq\omega$. In this case:
      \begin{align*}
        m_2(p) &= \mbar(p) &\textrm{By~(\ref{eq:9})}\\
        &= m_1(p)+\effect(\sigma_\pi)(p) &\textrm{By~(\ref{eq:10})}\\
        &= I(\sigma_\pi)(p)+\effect(\sigma_\pi)(p)+m(p) &\textrm{By~(\ref{eq:def-of-m1})}\\
        &\geq m(p)
      \end{align*}
    \item or $\effect(\sigma_\pi)(p)=\omega$. Then, $m_2(p)\geq m(p)$
      by~(\ref{eq:9})
    \end{enumerate}

  \item \emph{In case 2}, we know that $\effect(\mu(e))(p)\neq\omega$
    and $\effect(\sigmabar)(p)\neq\omega$, hence
    $\effect(\sigmabar\cdot\mu(e))\neq\omega$ and
    $\effect(\sigma_\pi)\neq\omega$ either. Then:
    \begin{align*}
      m_2(p) &= \mbar(p) &\textrm{By (\ref{eq:9})}\\
      &= m_1(p)+\effect(\sigma_\pi)(p) &\textrm{By~(\ref{eq:10})}\\
      &= \lambda(n_1)(p)+\effect(\sigma_\pi)(p) &\textrm{By~(\ref{eq:def-of-m1})}\\
      &=\lambda(n_2)(p) &\textrm{Lemma~\ref{lemma:no-omega-implies-exact-effect} and }\effect(\sigmabar\cdot\mu(e))\neq\omega
    \end{align*}
  \item \emph{In case 3}, $\lambda(n_2)(p)=\omega$ and
    $\effect(\sigma_\pi)(p)=\omega$ too. Hence, $m_2(p)\geq m(p)$ by
    (\ref{eq:9}).
  \item \emph{In case 4}, $\lambda(n_2)(p)=\omega$ again, and
    $m_1(p)=\lambda(n_1)(p)$, by~(\ref{eq:def-of-m1}). Moreover, we have
    $\effect(\sigma_\pi)(p)\neq\omega$, because
    $\effect(\sigmabar)(p)\neq\omega$ and
    $\effect(\mu(e))(p)\neq\omega$. Finally, since in case 4, we have
    $\effect(\sigmabar\cdot\mu(e))(p)>0$, and since
    $\sigma_\pi=\mu(e)\big(\sigmabar\cdot\mu(e)\big)^K$, we conclude that
    $\effect(\sigma_\pi)(p) \geq K-\effect(\mu(e))(p)$. Thus:
    \begin{align*}
      m_2(p) &= \mbar(p) &\textrm{By~(\ref{eq:9})}\\
      &=m_1(p)+\effect(\sigma_\pi)(p) &\textrm{By~(\ref{eq:10})}\\
      &\geq m_1(p) +K-\effect(\mu(e))(p)&\textrm{See above}\\
      &= m_1(p) + K-I(\mu(e))(p)+O(\mu(e))(p)&\textrm{Def. of }\effect\\
      &\geq K + m_1(p) - I(\mu(e))(p)\\
      &\geq K + \lambda(n_1)(p) - I(\mu(e))(p)&\textrm{By~(\ref{eq:def-of-m1})}\\
      &\geq K &\textrm{By~(\ref{eq:8})}\\
      &\geq m(p)&p\in P^{Acc}\textrm { and by~(\ref{eq:def-P-acc})
        and~(\ref{eq:def-sigma-pi})}
    \end{align*}\qed
  \end{itemize}
\end{proof}

We are now ready to prove
Lemma~\ref{lemma:path-implies-execution}:\bigskip

\noindent{\it
 Let $\Nn$ be an \wOPN, let $m_0$ be an \w-marking and let $\Tt$ be
  the tree returned by\break $\treealgo(\Nn,m_0)$. Let $\pi=n_0,\ldots, n_k$
  be a stuttering path in $\Tt$, and let $m$ be a marking in
  $\gamma(\lambda(n_k))$. Then, there exists an execution
  $\rho_\pi=m_0\xrightarrow{t_1}m_1\cdots\xrightarrow{t_\ell}m_\ell$
  of $\Nn$ s.t. $m_\ell\in\gamma(\lambda(n_k))$, $m_\ell\mgeq m$ and
  $m_0\in\gamma(\lambda(n_0))$.
  Moreover, when for all $0\leq i\leq j\leq k$: $\nbw{n_i}=\nbw{n_j}$,
  we have: $t_1\cdots t_\ell=\mu(\pi)$.
}

\begin{proof}
  We build, by induction on the length $k$ of the path in the tree, a
  corresponding execution of $\Nn$. The induction works backward,
  starting from the end of the path.
  
  \noindent\textbf{Base case, $k=0$}. Since $n_k=n_0$, we can take
  $m_0=m$, which clearly satisfies the Lemma since
  $m\in\lambda(n_k)=\lambda(n_0)$.

  \noindent\textbf{Inductive case, $k>0$}. The induction hypothesis is
  that there are a sequence of transitions $\sigma$ and two markings
  $m_1$ and $m_k$ s.t. $m_1\xrightarrow{\sigma}m_k$,
  $m_1\in\gamma(\lambda(n_1))$, $m_k\in\gamma(\lambda(n_k))$, and
  $m_k\geq m$. In the case where $(n_0,n_1)$ is not an edge of
  $\Tt$ (i.e., $n_1$ is an ancestor of $n_0$), we
  know that $\lambda(n_0)=\lambda(n_1)$ by definition of stuttering
  and let $\rho_{pi}= m_1\xrightarrow{\sigma}m_k$. Otherwise, we can
  apply Lemma~\ref{lemma:soundness}, and conclude that there are
  $\sigma'$, $m_0$ and $m_1'$ s.t. $m_0\xrightarrow{\sigma'}m_1'$,
  $m_0\in\gamma(\lambda(n_0))$, $m_1'\in\gamma(\lambda(n_1))$ and
  $m_1'\mgeq m_1$. Since $m_1'\mgeq m_1$, $\sigma$ is also firable
  from $m_1'$. Let $m_k'=m_1'+(m_k-m_1)$. Clearly,
  $m_0\xrightarrow{\sigma'}m_1'\xrightarrow{\sigma}m_k'$. Moreover,
  $m_k'\mgeq m_k\mgeq m$, by monotonicity. Let us show that
  $m_k'\in\gamma(\lambda(n_k))$. Since $m_1'$ and $m_1$ are both in
  $\gamma(\lambda(n_1))$: $m_1(p)=m_1'(p)$ for all $p$
  s.t. $\lambda(n_1)(p)\neq\omega$. Thus, by strong monotonicity, we
  conclude that $m_k(p)=m_k'(p)$ for all $p$
  s.t. $\lambda(n_1)(p)\neq\omega$. However, for all places $p$,
  $\lambda(n_k)(p)\neq\omega$ implies $\lambda(n_1)(p)\neq\omega$, as
  the number of $\omega$'s increase along a path in the tree. Thus we
  conclude that $m_k(p)=m_k'(p)$ for all $p$
  s.t. $\lambda(n_k)(p)\neq\omega$. Since $m_k(p)=\lambda(n_k)(p)$ for
  all $p$ s.t. $\lambda(n_k)(p)\neq\omega$ because
  $m_k\in\gamma(\lambda(n_k))$ by induction hypothesis, we conclude
  that $m_k'\in\gamma(\lambda(n_k))$ too. Thus, $m_0$, $m_k'$ and
  $\sigma'\cdot\sigma$ fulfill the statement of the lemma.

  Finally, observe that, when all the nodes along the path $\pi$ have
  the same number of $\omega$'s, Lemma~\ref{lemma:soundness}
  guarantees that $\mu(\pi)$ can be chosen for the sequence of
  transitions $\sigma$.\qed
\end{proof}
\section{Proof of
  Lemma~\ref{lemma:every-node-is-fully-developed-or-has-no-succ}}
\noindent\emph{Let $\Nn$ be an \wOPN, let $m_0$ be an \w-marking, and
  let $\Tt$ be the tree returned by $\treealgo(\Nn,m_0)$. Then, for
  all nodes $n$ of $\Tt$:
  \begin{itemize}
  \item either $n$ has no successor in the tree and has an ancestor
    $\nbar$ s.t. $\lambda(\nbar)=\lambda(n)$.
  \item or the set of successors of $n$ corresponds to all the
    $\rightarrow_\omega$ possible successors of $\lambda(n)$, i.e.:
    $\{\mu(n,n')\mid (n,n')\in E\} = \{t\mid
    \lambda(n)\xrightarrow{t}_\omega\}$. Moreover, for each $n'$
    s.t. $(n,n')\in E$ and $\mu(n,n')=t$:
    $\lambda(n')\mgeq\lambda(n)+\effect(t)$.
  \end{itemize}
}
\begin{proof}
  Observe that each time a node is created, it is inserted into
  $\mathtt{U}$, or a recursive call is performed on this node. In both
  cases, the node will eventually be considered in
  line~\ref{lst:if2}. If the condition of the \texttt{if} in
  line~\ref{lst:if2} is not satisfied, $n$ has an ancestor $\nbar$
  s.t. $\lambda(\nbar)=\lambda(n)$. Otherwise, all transitions $t$ that
  are firable from $\lambda(n)$ are considered in the loop in
  lines~\ref{lst:forall_t} onward, and a corresponding edge $(n,n')$
  with $\mu(n,n')=t$ is added to the tree in
  line~\ref{lst:connect}. The label $\lambda(n')$ of this node is
  either $\lambda(n)+\effect(t)$, or a $\mgeq$-larger marking, in the
  case where an acceleration has been performed during the
  \texttt{Post}, in line~\ref{lst:def-m-w}. Thus in both cases,
  $\lambda(n')\mgeq\lambda(n)+\effect(t)$. The algorithm terminates
  because \texttt{U} has become empty. Thus, all the nodes that have
  eventually been constructed by the algorithm fall into these two
  cases. Hence the Lemma.\qed
\end{proof}

\section{Proof of Lemma~\ref{lemma:execution-implies-path} (completeness)}
\noindent\textit{Let $\Nn$ be an \wOPN with set of transitions $T$,
  let $m_0$ be an initial marking, let $\Tt$ be the tree returned by
  $\treealgo(\Nn,m_0)$ and let
  $m_0\xrightarrow{t_1}m_1\xrightarrow{t_2}\cdots\xrightarrow{t_n}m_n$
  be an execution of $\Nn$. Then, there are a \emph{stuttering} path
  $\pi=n_0,n_1,\ldots,n_k$ in $\Tt$ and a monotonic increasing mapping
  $h:\{1,\ldots,n\}\mapsto\{0,\ldots,k\}$ s.t.: $\mu(\pi) =
  t_1t_2\cdots t_n$ and $m_i\mleq\lambda(n_{h(i)})$ for all $0\leq
  i\leq n$.}
\begin{proof}
  The proof is by induction on the length of the execution.

  \noindent\textbf{Base case: $n=0$} We let $h(0)=0$. By construction
  $\lambda(n_0)=m_0$, hence the lemma.

  \noindent\textbf{Inductive case: $n>0$} The induction hypothesis is
  that there are a path $\pi=n_0,\ldots n_\ell$ and a mapping
  $h:\{0,\ldots,n-1\}\mapsto \{0,\ldots,\ell\}$ satisfying the lemma
  for the execution prefix
  $m_0\xrightarrow{t_1}m_1\xrightarrow{t_2}\cdots\xrightarrow{t_{n-1}}m_{n-1}$.
  By Lemma~\ref{lemma:every-node-is-fully-developed-or-has-no-succ},
  we consider two cases for $n_\ell$:
  \begin{itemize}
  \item Either the set of successors of $n_\ell$ corresponds to the
    set of all transitions that are firable from
    $\lambda(n_\ell)$. Since, by induction hypothesis, $n_\ell\mgeq
    m_{n-1}$, and since $t_n$ is firable from $m_{n-1}$, we conclude
    that $t_n$ is firable from $\lambda(n_\ell)$ by
    monotonicity. Hence, $n_\ell$ has a successor $n$
    s.t. $\mu(n_\ell,n)=t_n$. Still by
    Lemma~\ref{lemma:every-node-is-fully-developed-or-has-no-succ},
    \begin{align*}
      \lambda(n) &\mgeq\lambda(n_\ell)+\effect(t_n)\\
      &\mgeq m_{n-1} +\effect(t_n)\\
      &\mgeq m_n\\
    \end{align*}
    Hence, we let $n_{\ell+1}=n$, and $h(n)=\ell+1$.
  \item Or the set of successors of $n_\ell$ is empty. In this case,
    by Lemma~\ref{lemma:every-node-is-fully-developed-or-has-no-succ},
    there exists an ancestor $n$ of $n_\ell$
    s.t. $\lambda(n)=\lambda(n_\ell)$. Let $n_{\ell+1}$ be such a
    node. Moreover, as $n_{\ell+1}\neq n_\ell$, and $n_{\ell+1}$ is an
    ancestor of $n_\ell$, $n_{\ell+1}$ must have at least one
    successor. Hence, by
    Lemma~\ref{lemma:every-node-is-fully-developed-or-has-no-succ},
    $n_{\ell+1}$ is fully developed, and we can apply the same
    reasoning as above to conclude that there is a successor $n'$ of
    $n_{\ell+1}$ s.t. $\lambda(n')\mgeq m_n$ and
    $\mu(n_{\ell+1},n')=t_n$. Let $n_{\ell+2}$ be such a node. We
    conclude by letting $h(n)=\ell+2$.\qed
  \end{itemize}
\end{proof}

\section{Proof of Lemma~\ref{lem-rem-I-w}}
\noindent{\it Let $\Nn$ be an \wPN. For all executions
  $m_0,t_1',m_1,\ldots, t_n',m_n$ of $\remIw(\Nn)$:
  $m_0,t_1,\break m_1,\ldots, t_n,m_n$ is an execution of $\Nn$. For all
  finite (resp. infinite) executions $m_0, t_1, m_1,\break\ldots, t_n, m_n$
  ($m_0, t_1, m_1,\ldots, t_j,m_j,\ldots$) of $\Nn$, there is an
  execution $m_0, t_1', m_1',\ldots,\break t_n',m_n'$ ($m_0, t_1,
  m_1',\ldots, t_j,m_j',\ldots$) of $\remIw(\Nn)$, s.t. $m_i\mleq
  m_i'$ for all $i$. }
\begin{proof}
  The first point follows immediately from the definition of
  $\remIw(\Nn)$ and from the fact that consuming $0$ tokens in each
  place $p$ s.t. $I(t_i)(p)=\omega$ is a valid choice when firing each
  transition $t_i$ in $\Nn$. The second point is easily shown by
  induction on the execution, because firing each $t_i$ produces the
  same amount of tokens that $t_i'$; consumes the same amount of token
  as each $t_i'$ in all places s.t. $I(t_i)(p)\neq \w$, and consumes,
  in each place $p$ s.t. $I(t_i)(p)=\w$ a number of tokens that is
  larger than or equal to the number of tokens consumed by
  $t_i'$.\qed
\end{proof}

\section{Proofs for Lemmas in Section~\ref{sec:extend-rack-techn}}
\begin{proof}[Lemma~\ref{lem:shortThPumpSeq}]
  This proof is similar to that of \cite[Lemma 4.5]{Rackoff78}, with
  some modifications to handle $\omega$-transitions. It is organized
  into the following steps.
  \begin{list}{Step}{}
    \item[Step 1:] We first associate a vector with a sequence of
      transitions to measure the effect of the sequence. This is the
      step that differs most from that of \cite[Lemma 4.5]{Rackoff78}.
      The idea in this step is similar to the one used in \cite[Lemma
      7]{BJK2010}.
    \item[Step 2:] Next we remove some simple loops from $\seq$ to
      obtain $\seq''$ such that for every intermediate
      $\omega$-marking $m$ in the run $m_{1} \tstep{\seq}{\threshold}
      m_{2}$, $m$ also occurs in the run $m_{1}
      \tstep{\seq''}{\threshold} m_{2}$.
    \item[Step 3:] The sequence $\seq''$ obtained above need not be a
      $\threshold$-PS. With the help of the vectors defined in step 1,
      we formulate a set of linear Diophantine equations that encode
      the fact that the effects of $\seq''$ and the simple loops that
      were removed in step 2 combine to give the effect of a
      $\threshold$-PS.
    \item[Step 4:] Then we use the result about existence of small solutions to
      linear Diophantine equations to construct a sequence $\seq'$
      that meets the length constraint of the lemma.
    \item[Step 5:] Finally, we prove that $\seq'$ is a $\threshold$-PS
      enabled at $m_{1}$.
  \end{list}
  
  \emph{Step 1}: Let $P_{\omega} \subseteq \cbasis{m_{1}}$ be the set
  of places $p$ such that some transition $t$ in $\seq$ has
  $\effect(t) ( p) = \omega$. If we ensure that for each place $p \in
  P_{\omega}$, some transition $t$ with $\effect (t) (p) = \omega$ is
  fired, we can ignore the effect of other transitions on $p$. This is
  formalized in the following definition of the effect of any sequence
  of transitions $\seq_{1} = t_{1} \cdots t_{r}$. We define the
  function $\effabs{P_{\omega}}{\seq_{1}}: \cbasis{m_{1}} \to \ZZ$ as
  follows.
  \begin{align*}
    \effabs{P_{\omega}}{\seq_{1}}(p) =
    \begin{cases}
      1 & p\in P_{\omega}, \exists i\in \{1, \ldots, r\}: \effect(t_{i})(p) = \omega\\
      0 & p\in P_{\omega}, \forall i\in \{1, \ldots, r\}: \effect(t_{i})(p) \ne \omega\\
      \sum_{1 \le i \le r} \effect ( t_{i})(p) & \text{otherwise}
    \end{cases}
  \end{align*}

  \emph{Step 2}: Let $m_{1} \tstep{\seq}{\threshold} m_{2}$. From
  Definition~\ref{def:thPumpSeq}, we have $\cbasis{m_{2}} =
  \cbasis{m_{1}}$. From Definition~\ref{def:thresholdMarkings}, infer
  that for any $\omega$-marking $m$ in the run $m_{1}
  \tstep{\seq}{\threshold} m_{2}$, $m(p) < \threshold( \nbwb{m_{1}})$
  for all $p \in P \setminus \cbasis{m_{1}}$.  Now we remove some
  simple loops from $\seq$ to obtain $\seq''$. To obtain some bounds
  in the next step, we first make the following observations on loops.
  Let $|P \setminus \cbasis{m_{1}}| = \numplaces$.
  Suppose $\seqbis$ is a simple loop. There can be at most
  $\threshold(\nbwb{m_{1}})^{\numplaces}$ transitions in $\seqbis$,
  so $-\threshold(\nbwb{m_{1}})^{\numplaces} \maxred \le
  \effabs{P_{\omega}}{\seqbis}(p) \le
  \threshold(\nbwb{m_{1}})^{\numplaces} \maxred$ for any $p \in P$.
  Let $\effs$ be the matrix whose set of columns is equal to
  $\{\effabs{P_{\omega}}{\seqbis} \mid \seqbis \text{ is a simple loop}
  \}$.  There are at most $(\threshold(\nbwb{m_{1}})^{\numplaces} 2
  \maxred)^{|P|}$ columns in $\effs$. We use $\effel, \effel', \ldots$
  to denote the columns of $\effs$.

  Now we remove simple loops from $\seq$ according to the
  following steps. Let $\lvv_{0} = \0$ be the zero vector whose
  dimension is equal to the number of columns in $\effs$. Begin the
  following steps with $i = 0$ and $\seq_{i} = \seq$.
  \begin{enumerate}
      \renewcommand{\theenumi}{\alph{enumi}}
    \item Think of the first $(\threshold(\nbwb{m_{1}})^{|P|} +
      1)^{2}$ transitions of $\seq_{i}$ as
      $\threshold(\nbwb{m_{1}})^{|P|} + 1$ blocks of length
      $\threshold(\nbwb{m_{1}})^{|P|} + 1$ each.
    \item There is at least one block in which all
      $\omega$-markings also occur in some other block.
    \item Let $\seqbis$ be a simple loop occurring
      in the above block.
    \item Let $\seq_{i + 1}$ be the sequence obtained from
      $\seq_{i}$ by removing $\seqbis$.
    \item Let $\lvv_{i + 1}$ be the vector obtained from
      $\lvv_{i}$ by incrementing $\lvv_{i}
      (\effabs{P_{\omega}}{\seqbis})$ by $1$.
    \item Increment $i$ by $1$.
    \item If the length of the remaining sequence is more than or
      equal to $(\threshold(\nbwb{m_{1}})^{|P|}\break + 1)^{2}$,  go back to step
      a. Otherwise, stop.
  \end{enumerate}
  Let $n$ be the value of $i$ when the above process stops. Let
  $\seq'' = \seq_{n}$ and $\lvv = \lvv_{n}$. We remove a simple loop
  $\seqbis$ starting at an $\omega$-marking $m$ only if all the
  intermediate $\omega$-markings occurring while firing $\seqbis$ from
  $m$ occur at least once more in the remaining sequence.  Hence, for
  every $\omega$-marking $m$ arising while while firing $\seq$ from
  $m_{1}$, $m$ also arises while firing $\seq''$ from $m_{1}$. We have
  $|\seq''| \le (\threshold(\nbwb{m_{1}})^{|P|} + 1)^{2}$. For
  each column $\effel$ of $\effs$, $\lvv ( \effel)$ contains the
  number of occurrences of simple loops $\seqbis$ removed from $\seq$
  such that $\effabs{P_{\omega}}{\seqbis} = \effel$.

  \emph{Step 3}: For every $p \in P_{\omega}$, we want to ensure that
  there is some transition $t$ in the shorter $\threshold$-PS that we
  will build, such that $\effect (t ) (p) = \omega$. For the other
  places, we want to ensure that the effect of the shorter
  $\threshold$-PS is non-negative. These requirements are expressed in
  the following vector $\fev$.
  \begin{align*}
    \fev(p) =
    \begin{cases}
      1 & p \in P_{\omega}\\
      0 & p \notin P_{\omega}
    \end{cases}
  \end{align*}
  Recall that for each column $\effel$ of $\effs$, $\lvv ( \effel)$
  contains the number of occurrences of simple loops $\seqbis$ removed
  from $\seq$ such that $\effabs{P_{\omega}}{\seqbis} = \effel$ and
  that $\seq''$ is the sequence remaining after all removals. Hence,
  $\effabs{P_{\omega}}{\seq} = \effs \lvv +
  \effabs{P_{\omega}}{\seq''}$.  Since $\seq$ is a $\threshold$-PS and
  for every $p \in P_{\omega}$, there is a transition $t$ in $\seq$
  such that $\effect (t ) ( p) = \omega$, we have
  \begin{align}
    \nonumber \effabs{P_{\omega}}{\seq}  & \ge \fev\\
    \nonumber \Rightarrow \effs \lvv + \effabs{P_{\omega}}{\seq''}  & \ge \fev\\
    \Rightarrow \effs \lvv & \ge \fev - \effabs{P_{\omega}}{\seq''}
    \enspace .
    \label{eq:TotalEff}
  \end{align}

  \emph{Step 4:} We use the following result about the existence of
  small integral solutions to linear equations \cite{BT76},
  which has been used by Rackoff to give \expsp{} upper bound for the
  boundedness problems in Petri nets \cite[Lemma 4.4]{Rackoff78}.

  \begin{em}
    Let $d_{1}, d_{2} \in \NNP$, let $\vec{A}$ be a $d_{1} \times
    d_{2}$ integer matrix and let $\vec{a}$ be an integer vector of
    dimension $d_{1}$. Let $d \ge d_{2}$ be an upper bound on the
    absolute value of the integers in $\vec{A}$ and $\vec{a}$. Suppose
    there is a vector $\vec{x} \in \NN^{d_{2}}$ such that $\vec{A}
    \vec{x} \ge \vec{a}$. Then for some constant $c$ independent of $d,
    d_{1}, d_{2}$, there exists a vector $\vec{y} \in \NN^{d_{2}}$
    such that $\vec{A} \vec{y} \ge \vec{a}$ and $\vec{y} ( i) \le
    d^{c d_{1}}$ for all $i$ between $1$ and $d_{2}$.
  \end{em}

  We apply the above result to \eqref{eq:TotalEff}. Each entry of
  $\effabs{P_{\omega}}{\seq''}$ is of absolute value at most
  $(\threshold(\nbwb{m_{1}})^{|P|} + 1)^{2} \maxred$.  Recall that
  there are at most $(\threshold(\nbwb{m_{1}})^{\numplaces} 2
  \maxred)^{|P|}$ columns in $\effs$, with the absolute value of each
  entry at most $\threshold(\nbwb{m_{1}})^{\numplaces} \maxred$. There
  are $|P| - \numplaces$ rows in $\effs$. Hence, we conclude that
  $\lvv$ can be replaced by $\lvvbis$ such that $\effs \lvvbis \ge
  \fev - \effabs{P_{\omega}}{\seq''}$ and the sum of all entries in
  $\lvvbis$ is at most $(\threshold(\nbwb{m_{1}}) 2 \maxred)^{d'
    |P|^{3}}$ for some constant $d'$.  This expression is obtained
  from simplifying
  $$
  (\threshold(\nbwb{m_{1}})^{\numplaces} 2 \maxred)^{|P|}
  ((\threshold(\nbwb{m_{1}})^{|P|} + 1)^{2} 2 \maxred)^{d''|P|^{2}}
  $$ 
  for some constant $d''$.

  For each column $\effel$ of $\effs$, let $\seqbis_{\effel}$ be a
  simple loop of $\seq$ such that
  $\effabs{P_{\omega}}{\seqbis_{\effel}} = \effel$. Recall from step 2
  that there is some intermediate $\omega$-marking $m_{\effel}$
  occurring while firing $\seq''$ from $m_{1}$ such that $m_{\effel}$
  is the $\omega$-marking from which the simple loop
  $\seqbis_{\effel}$ is fired in $\seq$. Let $i_{\effel}$ be the
  position in $\seq''$ where $m_{\effel}$ occurs. Let $\seq'$ be the
  sequence obtained from $\seq''$ by inserting $\lvvbis ( \effel)$
  copies of $\seqbis_{\effel}$ into $\seq''$ at the position
  $i_{\effel}$ for each column $\effel$ of $\effs$. Since we insert at
  most $(\threshold(\nbwb{m_{1}}) 2 \maxred)^{d' |P|^{3}}$ simple
  loops, each of length at most
  $\threshold(\nbwb{m_{1}})^{\numplaces}$, $|\seq'| \le
  (\threshold(\nbwb{m_{1}}) 2 \maxred)^{d' |P|^{3}}
  \threshold(\nbwb{m_{1}})^{\numplaces} +
  (\threshold(\nbwb{m_{1}})^{|P|} + 1)^{2}$.  Choose the constant $d$
  s.t.  $|\seq'| \le  (\threshold(\nbwb{m_{1}}) 2 \maxred)^{d'
  |P|^{3}}\times\break \threshold(\nbwb{m_{1}})^{\numplaces} +
  (\threshold(\nbwb{m_{1}})^{|P|} + 1)^{2} \le (
  \threshold(\nbwb{m_{1}}) 2 \maxred)^{d |P|^{3}}$. Now we have
  $|\seq'| \le (\threshold(\nbwb{m_{1}}) 2 \maxred)^{d |P|^{3}}$.

  \emph{Step 5:} Now we prove that $\seq'$ is a $\threshold$-PS
  enabled at $m_{1}$. Recall that $m_{1} \tstep{ \seq}{ \threshold}
  m_{2}$ and that $\seq'$ is obtained from $\seq$ by removing or
  adding extra copies of some simple loops. We infer that $m_{1}
  \tstep{\seq'}{\threshold} m_{2}$. Now we show that $\effect ( \seq')
  \mgeq \0$. Since for any simple loop $\seqbis$ in $\seq$, $\effect (
  \seqbis) ( p) = 0$ for all $p \in P \setminus \cbasis{m_{1}}$, we
  have $\effect( \seq') ( p) = \effect( \seq) ( p) \ge 0$.

  For any $p \in P_{\omega}$, we have $(\effs \lvvbis +
  \effabs{P_{\omega}}{\seq''}) (p) \ge \fev( p) \ge 1$.  Hence,
  $\lvvbis ( \effabs{P_{\omega}}{\seqbis}) \ge 1$ and
  $\effabs{P_{\omega}}{\seqbis} ( p) = 1$ for some simple loop
  $\seqbis$ or $\effabs{P_{\omega}}{\seq''} (p) = 1$. From the
  definitions of $\effabs{P_{\omega}}{\seqbis}$ and
  $\effabs{P_{\omega}}{\seq''}$, the only way this can happen is for
  some transition $t$ in either some simple loop $\seqbis$ or $\seq''$
  to have $\effect ( t) = \omega$.  Hence, there is some transition
  $t$ in $\seq'$ such that $\effect( t) ( p) = \omega$. Hence,
  $\effect( \seq') ( p) = \omega$.

  For any $p \in \cbasis{m_{1}} \setminus P_{\omega}$, we have
  $\effect( \seq') ( p) = (\effs \lvvbis +
  \effabs{P_{\omega}}{\seq''}) (p) \ge \fev( p) \ge 0$. Hence,
  $\effect( \seq') ( p) \ge 0$. \qed
\end{proof}

\begin{proof}[Lemma~\ref{lem:thCovToNatCov}]
  Let $\seq'$ be obtained from $\seq$ by removing all transitions
  between any two identical $\omega$-markings occurring in the run
  $m_{3} \tstep{\seq}{\threshold_{1}} m_{4}$. The number of distinct
  $\omega$-markings appearing in the run $m_{3}
  \tstep{\seq'}{\threshold_{1}} m_{4}$ is an upper bound on $|\seq'|$.
  Among the $\omega$-markings in this run, $m_{3}$ has the maximum
  number of places not marked $\omega$. Since $\threshold_{1}$ is
  non-decreasing, we infer from the definition of threshold semantics
  (Definition~\ref{def:thresholdPN}) that $\threshold_{1}
  (\nbwb{m_{3}} )^{|P|}$ is an upper bound on the number of  possible
  distinct $\omega$-markings. Hence, $|\seq'| \le \threshold_{1}
  (\nbwb{m_{3}} )^{|P|}$. We will now prove that for any run $m_{3}
  \tstep{\seq'}{\threshold_{1}} m_{4}$ where all intermediate
  $\omega$-markings are distinct from one another,
  $\floorthr{m_{3}}{\threshold_{1}} \step{\seq'} m_{4}'$ and
  $m_{4}' \mgeqp{\cbasis{m_{4}}} \floorthr{m_{4}}{\threshold_{2}}$.
  The proof is by induction on $\nbw{m_{4}} - \nbw{m_{3}}$ (the number
  of places where $\omega$ is newly introduced).

  \emph{Base case $\nbw{m_{4}} - \nbw{m_{3}} = 0$}:
  We have $|\seq'| \le
  \threshold_{1}( \nbwb{m_{3}})^{|P|} \le \lenfn( \nbwb{m_{3}})$. For
  any $p' \in \cbasis{m_{3}}$, we have by
  Definition~\ref{def:thresholdMarkings} and
  Definition~\ref{def:ThFns} that $\floorthr{m_{3}}{\threshold_{1}}
  (p') = \threshold_{1} ( \nbwb{m_{3}} + 1) = 2 \maxred \lenfn(
  \nbwb{m_{3}})$. We conclude from Proposition
  \ref{prop:thSemToNatSem} that $\floorthr{m_{3}}{\threshold_{1}}
  \step{\seq'} m_{4}'$ and $m_{4}' \mgeqp{\cbasis{m_{4}}}
  \floorthr{m_{4}}{\threshold_{2}}$.

  \emph{Induction step}: Let $m_{5}$ be the first $\omega$-marking
  after $m_{3}$ such that $\nbw{m_{5}} > \nbw{m_{3}}$. Let $\seq' =
  \seq_{1} t \seq_{2}$ where $m_{3} \tstep{\seq_{1}}{\threshold_{1}}
  m_{6} \tstep{t}{\threshold_{1}} m_{5}
  \tstep{\seq_{2}}{\threshold_{1}} m_{4}$. Note that due to our choice
  of $m_{5}$, we have $\cbasis{m_{6}} = \cbasis{m_{3}}$.  In any
  intermediate marking $m \ne m_{3}$ in the run $m_{3}
  \tstep{\seq_{1}}{\threshold_{1}} m_{6}$, $m(p) < \threshold_{1}(
  \nbwb{m_{3}})$ for all $p \in P \setminus \cbasis{m_{3}}$
  (otherwise, $p$ would have been marked $\omega$, contradicting
  $\cbasis{m_{6}} = \cbasis{m_{3}}$). Hence we have $|\seq_{1}| \le
  \threshold_{1}( \nbwb{m_{3}})^{|P|}$. For any $p' \in
  \cbasis{m_{3}}$, we have by Definition~\ref{def:thresholdMarkings}
  and Definition~\ref{def:ThFns} that
  $\floorthr{m_{3}}{\threshold_{1}} (p') = \threshold_{1} (
  \nbwb{m_{3}} + 1) = 2 \maxred \lenfn( \nbwb{m_{3}})$. We conclude
  from Proposition~\ref{prop:thSemToNatSem} that
  $\floorthr{m_{3}}{\threshold_{1}} \step{\seq_{1}} m_{6}'$ where
  $m_{6}' \mleqp{\cbasis{m_{6}}} m_{6}$ and for all $p' \in
  \cbasis{m_{6}}$, $m_{6}' (p') \ge 2 \maxred \lenfn( \nbwb{m_{3}}) -
  \maxred \threshold_{1}( \nbwb{m_{3}})^{|P|}$. Transition $t$ is
  enabled at $m_{6}'$. Let $m_{6}' \step{t} m_{5}'$, where for any $p$
  such that $\effect(t) (p) = \omega$, we chose $m_{5}' ( p) \ge
  \threshold_{1} ( \nbwb{m_{5}} + 1)$. We now conclude that
  $m_{5}' \mgeqp{\cbasis{m_{5}}} \floorthr{m_{5}}{\threshold_{1}}$ due
  to the following reasons:
  \begin{enumerate}
    \item $p \in P \setminus \cbasis{m_{5}}$: we have $p \in P
      \setminus \cbasis{m_{6}}$.
      \begin{align*}
	m_{5}' ( p) & = m_{6}' ( p) + \effect(t) & \text{[semantics of
	 \wPN]}\\
	& = m_{6} (p) + \effect(t) & [m_{6}' \mleqp{\cbasis{m_{6}}}
	m_{6}]\\
	& = m_{5} (p) & [\ceilthr{m_{6} +
	\effect(t)}{\threshold_{1}} = m_{5}, m_{5}(p) \ne \omega]\\
	& = \floorthr{m_{5}}{\threshold_{1}} (p)
      \end{align*}
    \item $p \in \cbasis{m_{5}}$, $\effect (t) (p) = \omega$:
      $m_{5}' ( p) \ge \threshold_{1} ( \nbwb{m_{5}} + 1)$ by choice.
    \item $p \in \cbasis{m_{5}}$, $\effect (t) (p) \ne \omega$,
      $p \notin \cbasis{m_{6}}$: since $\ceilthr{m_{6} +
      \effect(t)}{\threshold_{1}} = m_{5}$ and $m_{5} ( p) = \omega$,
      \begin{align*}
	m_{6} (p) + \effect(t)(p) & \ge \threshold_{1} (
	\nbwb{m_{6}})\\
	\Rightarrow m_{6} (p) + \effect(t)(p) & \ge \threshold_{1} (
	\nbwb{m_{5}} + 1) & [\nbw{m_{5}} > \nbw{m_{6}}]\\
	\Rightarrow m_{6}' (p) + \effect(t)(p) & \ge \threshold_{1} (
	\nbwb{m_{5}} + 1) & [m_{6}' \mleqp{\cbasis{m_{6}}}
	m_{6}]\\
	\Rightarrow m_{5}' ( p) & \ge \threshold_{1} (
	\nbwb{m_{5}} + 1) & \text{[semantics of
	 \wPN]}
      \end{align*}
    \item $p \in \cbasis{m_{5}}$, $\effect (t) (p) \ne \omega$,
      $p \in \cbasis{m_{6}}$:
      \begin{align*}
	m_{5}' ( p) & = m_{6}' ( p) + \effect(t)(p) & \text{[semantics of
	 \wPN]}\\
	 & \ge m_{6}' ( p) - \maxred & \text{[Definition of } \maxred
	 \text{]}\\
	 & \ge 2 \maxred \lenfn( \nbwb{m_{3}}) - \maxred
	 \threshold_{1}( \nbwb{m_{3}})^{|P|} - \maxred & [p \in
	 \cbasis{m_{6}}]\\
	 & \ge \maxred \lenfn( \nbwb{m_{3}}) - \maxred
	 \threshold_{1}( \nbwb{m_{3}})^{|P|}\\
	 & = \maxred (\threshold_{1}(\nbwb{m_{3}}) 2
	 \maxred)^{c|P|^{3}} - \maxred \threshold_{1}(
	 \nbwb{m_{3}})^{|P|} & [\text{Definition~\ref{def:ThFns}}]\\
	 & \ge \threshold_{1}(\nbwb{m_{3}})\\
	 & \ge \threshold_{1}(\nbwb{m_{5}} + 1)
      \end{align*}
      The last inequality follows since $\nbw{m_{5}} > \nbw{m_{3}}$.
  \end{enumerate}

  Since $\nbw{m_{4}} - \nbw{m_{5}} < \nbw{m_{4}} - \nbw{m_3}$ and all
  intermediate $\omega$-markings in the run $m_{5}
  \tstep{\seq_{2}}{\threshold_{1}} m_{4}$ are distinct from one
  another, we have by induction hypothesis that
  $\floorthr{m_{5}}{\threshold_{1}} \step{\seq_{2}} m_{4}''$ and
  $m_{4}'' \mgeqp{\cbasis{m_{4}}} \floorthr{m_{4}}{\threshold_{2}}$.
  Since $\floorthr{m_{3}}{\threshold_{1}} \step{\seq_{1}} m_{6}'
  \step{t} m_{5}'$, $m_{5}' \mgeqp{\cbasis{m_{5}}}
  \floorthr{m_{5}}{\threshold_{1}}$ and
  $\floorthr{m_{5}}{\threshold_{1}} \step{\seq_{2}} m_{4}''$, we infer
  by strong monotonicity that $\floorthr{m_{3}}{\threshold_{1}}
  \step{\seq_{1}t\seq_{2}} m_{4}'$ and
  $m_{4}' \mgeqp{\cbasis{m_{4}}}
  \floorthr{m_{4}}{\threshold_{2}}$.\qed
\end{proof}

\begin{proof}[Lemma~\ref{lem:GrowthFnUpBnd}]
  By induction on $i$. For the base case $i = 0$, the result is
  obvious since by Definition~\ref{def:ThFns}, $\lenfn( 0) = (
  2\maxred)^{c |P|^{3}}$.

  Induction step:
  \begin{align*}
    \lenfn( i  + 1 ) & = (\threshold_{1}( i + 1 ) 2 \maxred)^{ c
    |P|^{3}} & \text{[Definition~\ref{def:ThFns}]}\\
    & = (2 \maxred \lenfn(i) \cdot 2 \cdot \maxred)^{ c
    |P|^{3}} & \text{[Definition~\ref{def:ThFns}]}\\
    & = (4 \maxred^{2})^{c |P|^{3}}
    (\lenfn(i))^{c |P|^{3}}\\
    & = (2 \maxred)^{2 c |P|^{3}} (\lenfn(i))^{c
    |P|^{3}}\\
    & \le (2 \maxred)^{2 c |P|^{3}} (( 2 \maxred)^{
    k^{i + 1} |P|^{3 ( i + 1)}})^{ c
    |P|^{3}} & \text{[Induction hypothesis]}\\
    & = (2 \maxred)^{2 c |P|^{3}} ( 2 \maxred)^{
    c k^{i + 1} |P|^{3 ( i + 2)}}\\
    &\le (2 \maxred)^{3 c k^{i + 1} |P|^{3
    ( i + 2)}}\\
    & = (2 \maxred)^{k^{i + 2} |P|^{3
    ( i + 2)}}
  \end{align*}\qed
\end{proof}

\begin{proof}[Theorem~\ref{thm:terminationExpsp}]
  Since \wPN generalise Petri nets, and since termination is
  \textsc{ExpSpace}-c for Petri nets
  \cite{Rackoff78}\todo{Check Ref.}, termination is
  \textsc{ExpSpace}-hard for \wPN. Let us now show that termination
  for \wPN is in \textsc{ExpSpace}. We have from
  Lemma~\ref{lemma:terminates-iff-no-s-c-e} that an \wPN \Nn does not
  terminate iff it admits a self-covering execution. From
  Lemma~\ref{lem:shortSCS}, it admits a self-covering execution iff it
  admits one whose sequence of transitions is of length at most
  $\lenfn( |P|)$. The following non-deterministic algorithm can guess
  and verify the existence of such a sequence. It works with
  $\omega$-markings, storing $\omega$ in the respective places whenever
  an $w$-transition is fired.
  \begin{description}
    \item[Input] An \wPN \Nn, with initial marking $m_{0}$.
    \item[Output] SUCCESS if a self-covering execution is guessed,
      FAIL otherwise.
  \end{description}
  \begin{pseudo}
    counter := $0$
    $m$ := $m_{0}$
    if counter $ > \lenfn(|P|) $ @\label{firstLoop}@
    	return FAIL
    else
    	non-deterministically choose a transition $t$
	if $t$ is not enabled at $m$
		return FAIL
	else
		$m$ := $m + \effect(t)$
		counter := counter + $1$
	non-deterministically go to line@~\ref{firstLoop}@ or line@~\ref{beginPS}@ @\label{endFirstLoop}@
    in $m$, replace $\omega$ by $\maxred \lenfn(|P|)$ @\label{beginPS}@
    $m_{1}$ := $m$
    if counter $ > \lenfn(|P|)$ @\label{secondLoop}@
    	return FAIL
    else
    	non-deterministically choose a transition $t$
	if $t$ is not enabled at $m_{1}$
		return FAIL
	else
		$m_{1}$ : = $m_{1} + \effect(t)$
		counter := counter + $1$
		non-deterministically go to line@~\ref{secondLoop}@ or line@~\ref{endPS}@ @\label{endSecondLoop}@
    if $m_{1} \mgeq m$ @\label{endPS}@
    	return SUCCESS
    else
    	return FAIL
  \end{pseudo}
  The above algorithm tries to guess a sequence of transitions
  $\seq_{1} \seq_{2}$ such that $m_{0} \step{\seq_{1}} m
  \step{\seq_{2}} m_{1}$, guessing $\seq_{1}$ in the loop between
  lines \ref{firstLoop} and \ref{endFirstLoop} and $\seq_{2}$ in the loop
  between lines \ref{secondLoop} and \ref{endSecondLoop}. If \Nn
  admits a self-covering execution with sequence of transitions
  $\seq_{1} \seq_{2}$ such that $| \seq_{1} \seq_{2} | \le \lenfn
  ( |P| )$, then the execution of the above algorithm that guesses
  $\seq_{1} \seq_{2}$ will return SUCCESS. If all executions of \Nn
  are finite, then all executions of the above algorithm will return
  FAIL.

  The space required to store the variable ``counter'' in the above
  algorithm is at most $\log ( \lenfn( |P|) )$. The space required to
  store $m$ and $m_{1}$ is at most $|P|(\infnorm{m_{0}} + \log (
  \maxred \lenfn( |P|)))$.  Using the upper bound given by
  Lemma~\ref{lem:GrowthFnUpBnd}, we conclude that the memory space
  required by the above algorithm is $\Oh (|P| \log \infnorm{ m_{0}} +
  k^{ |P| + 1} |P|^{ 3 |P| + 4} \log \maxred)$.  This can be
  simplified to $\Oh( 2^{c' |P| \log |P|} (\log \maxred + \log
  \infnorm{m_{0}}))$. Using the well known Savitch's theorem to
  determinize the above algorithm, we get an \expsp{} upper bound for
  the termination problem in \wPN.\qed
\end{proof}

\end{document}